\documentclass[onefignum,onetabnum,dvipsnames]{siamart220329}

\usepackage{lipsum}
\usepackage{amsfonts}
\usepackage{graphicx}
\usepackage{epstopdf}
\usepackage{algorithm}
\usepackage{subdepth}
\usepackage{algpseudocode}
\ifpdf
  \DeclareGraphicsExtensions{.eps,.pdf,.png,.jpg}
\else
  \DeclareGraphicsExtensions{.eps}
\fi

\usepackage{mathtools}
\usepackage{nicefrac}
\usepackage{cleveref}
\usepackage{enumitem}

\newcommand{\eps}{\varepsilon}
\newcommand*{\opt}{\textsc{Opt}}
\newcommand*{\alg}{\textsc{Alg}}

\newcommand{\N}{\mathbb{N}}
\newcommand{\length}{\ell}
\newcommand{\cost}{c}
\newcommand{\Length}{L}
\newcommand{\obj}{total latency}
\renewcommand{\P}{\mathsf{P}}
\newcommand{\NP}{\mathsf{NP}}

\usepackage{xspace}
\usepackage{subcaption}
\usepackage{tikz}
\tikzstyle{state}=
[circle,draw=black,thick,fill=black,minimum size=1.5mm,inner sep=0pt]

\definecolor{myyellow}{cmyk}{0, 0.45, 0.90, 0.17}
\newcommand{\myblue}{MidnightBlue}
\newcommand{\myyellow}{myyellow}
\newcommand{\mygreen}{LimeGreen}
\newcommand{\myred}{Red}

\usepackage{todonotes}

\makeatletter%
\@mparswitchfalse%
\makeatother%
\normalmarginpar

\newsiamremark{remark}{Remark}
\newsiamremark{hypothesis}{Hypothesis}
\crefname{hypothesis}{Hypothesis}{Hypotheses}
\newsiamthm{claim}{Claim}

\headers{Expanding Search Problem}{S.~M.~Griesbach, F.~Hommelsheim, M.~Klimm, and K.~Schewior}

\title{Improved Approximation Algorithms\\ for the Expanding Search Problem\thanks{
The results in this paper have been published in preliminary form in the Proceedings of the 31st Annual European Symposium on Algorithms (ESA)~\cite{griesbach2024improved}.
\funding{This work was supported by
Deutsche Forschungsgemeinschaft under Germany’s Excellence Strategy, Berlin Mathematics Research Center (grant EXC-2046/1, Project 390685689), 
Centro de Modelamiento Matemático (CMM) BASAL fund FB210005 for center of excellence from ANID-Chile,
HYPATIA.SCIENCE (Department of Mathematics and Computer Science, University of Cologne,
and in part by the Independent Research Fund Denmark, Natural Sciences (grant DFF-0135-00018B).}}}

\author{Svenja~M.~Griesbach\thanks{Centro de Modelamiento Matemático (CNRS IRL2807), Universidad de Chile, Santiago, Chile (\email{sgriesbach@cmm.uchile.cl},
\url{https://sites.google.com/view/svenja-m-griesbach}).} 
\and Felix~Hommelsheim%
\thanks{Universität Bremen, Faculty of Mathematics and Computer Science 
  (\email{fhommels@uni-bremen.de}, \url{https://www.uni-bremen.de/en/cslog/team/felix-hommelsheim}).}
\and Max~Klimm%
\thanks{Technische Universität Berlin, Institute for Mathematics 
  (\email{klimm@tu-berlin.de}, \url{https://tu.berlin/disco/klimm}).}
\and Kevin~Schewior\thanks{University of Southern Denmark, Department of Mathematics and Computer Science (\email{kevs@sdu.dk}, \url{https://sites.google.com/view/kschewior/})}.}

\usepackage{amsopn}

\ifpdf
\hypersetup{
  pdftitle={Improved Approximation Algorithms for the Expanding Search Problem},
  pdfauthor={S.~M.~Griesbach, F.~Hommelsheim, M.~Klimm, and K.~Schewior}
}
\fi

\begin{document}

\maketitle

\begin{abstract}
A searcher is tasked with exploring a graph with edge lengths and vertex weights, starting from a designated vertex. Initially, only the starting vertex is considered explored. At each step, the searcher adds an edge to the solution, connecting an unexplored vertex to an explored one. The time required to add an edge equals its length. The objective is to minimize the weighted sum of exploration times for all vertices.

We demonstrate that this problem is hard to approximate and present algorithms with improved approximation guarantees. Specifically, we provide a~$(2\mathrm{e} + \varepsilon)$-approximation for any~$\varepsilon > 0$ for the general case.
On instances where the vertex weights are binary, we achieve a~$2\mathrm{e}$-approximation. Finally,  we develop a polynomial-time approximation scheme (PTAS) for Euclidean graphs. 
Previously, only an~$8$-approximation was known for all these cases.
\end{abstract}



\section{Introduction}

A critical challenge faced by disaster-relief teams deployed to regions devastated by natural or human-made catastrophes is determining where to search for buried or isolated people. The fundamental issues behind these decisions are that, in emergencies, technical means for probing and clearing areas are often limited, there is incomplete knowledge regarding the exact locations of potential survivors, and rescue operations are time-sensitive since the chances of survival decrease with the time taken for rescue. 

Following Averbakh and Pereira~\cite{averbakh2012flowtime},  this problem is modeled using an undirected graph with edge lengths. The graph's vertices represent different locations within the disaster area, and the edges represent possible connections between these locations.
The length of an edge corresponds to the time required to clear the connection, which could involve removing rubble from a road, defusing explosives, or digging through snow, dirt, or debris.
A single rescue team starts at a designated root vertex. Based on prior knowledge, the rescue team is aware of the number of survivors at each location, which is represented by the vertex weights. The objective is to minimize the average time it takes for the team to reach a survivor. 

A solution to the problem is represented by a sequence of edges to clear until all vertices with non-zero weight can be reached. Once an edge is cleared, the rescue team can travel along it in negligible time, so we only consider the time required to clear the edges. This type of problem is referred to as an \emph{expanding search problem}, where the set of vertices accessible by the rescue team expands with each step. This contrasts with \emph{pathwise search problems}, where the movement of the searcher is modeled in such a way that traversing an edge always takes time equal to the edge length, regardless of whether it is the first traversal or not.

Expanding search problems are particularly suitable when the time required to traverse an edge for the first time is significantly higher than the time needed for subsequent traversals, allowing the time for later movements to be neglected.
Beyond rescue operations (Averbakh and Pereira~\cite{averbakh2012flowtime}), such problems also arise in mining where the time to dig a new tunnel is much higher than the time to traverse already-dug tunnels to previously explored locations (Alpern and Lidbetter~\cite{alpern2013mining}), and when securing an area from hidden explosives where the time to move within a safe region is negligible compared to the time required to secure a new area (Angelopoulos et al.~\cite{angelopoulos2019expanding}).

\subsection{Our Contribution}
In this work, we provide polynomial-time approximation al\-go\-rithms with improved approximation guarantees for the expanding search problem.
Specifically, we give an approximation algorithm for arbitrary vertex weights with an approximation guarantee of~$2\mathrm{e}+\eps$ for any~$\eps > 0$ where~$2\mathrm{e} \approx 5.44$ (\Cref{thm:weighted}).
To this end, we first devise an approximation algorithm for the unweighted case where all vertices have the same weight, for which we obtain a~$2\mathrm{e}$-approximation (Theorem~\ref{thm:unweighted}).
This algorithm is obtained by concatenating~$k$-minimum spanning trees ($k$-MSTs) for varying values of~$k$ and of exponentially increasing length.
Using the~$2$-approximation for~$k$-MST due to Garg~\cite{Garg05} causes the loss of a factor of~$2$; using the probabilistic method to find the concatenation of the right trees, we incur an additional loss of a factor of~$\mathrm{e}$.
A similar technique has been used before for pathwise search problems~\cite{chaudhuri2003paths,goemans1998improved}; we here modify it to work for expanding search problems.
To obtain the approximation algorithm for the case with arbitrary vertex weights, we adapt an approach developed by Sitters~\cite{sitters2021polynomial} to obtain a polynomial-time approximation scheme (PTAS) for the pathwise search problem in the case of a Euclidean graph. We describe this adapted approach in more detail below as we also use it to obtain a PTAS for the case of a Euclidean graph. In our context, it allows us to reduce the weighted case to the case with binary weights (i.e., weights~$0$ and~$1$) at the cost of an additional factor of~$(1+\eps)$.
The application of this approach requires a~$2\mathrm{e}$-approximation for binary weights which we obtain by a generalization of our~$2\mathrm{e}$-approximation algorithm for the unweighted case.

We then provide the aforementioned PTAS for the case of a Euclidean graph (Theorem~\ref{thm:PTAS}) where vertices correspond to points in the plane, all pairs of points are connected by an edge, and the length of that edge is equal to the Euclidean distance of their endpoints. The approach by Sitters~\cite{sitters2021polynomial} for the pathwise search problem relies on partitioning an instance into subinstances. A central challenge when adapting this approach to the expanding search setting is that, unlike pathwise search, expanding search is not memoryless, as points contained in one subinstance may be used as Steiner points in another subinstance. We address this difficulty by keeping such points in the subinstance with zero weight so that the partitions overlap; hence the case of binary weights emerges. To obtain a PTAS for the subproblem, we adapt techniques developed by Arora~\cite{Arora98} for the Euclidean travelling salesperson problem. 

Prior to our work, for all variants considered in this paper, i.e., the unweighted case, the weighted case, and the Euclidean case, the best approximation algorithm was an~$8$-approxi\-mation due to Hermans et al.~\cite{HermansLM22}. 

Finally, we show that there is no PTAS for the expanding search problem on general graphs unless~$\mathsf{P} = \mathsf{NP}$ (Theorem~\ref{thm:hardness}).
In light of our PTAS for Euclidean instances, this result implies that expanding search is considerably more difficult on general graphs than on Euclidean graphs.
The proof follows a similar idea as the hardness proof for the traveling repairperson problem suggested in~\cite{blum1994minimum}.
However, the solution needs to be more structured in our setting compared to the pathwise search. Showing this property turns out to be rather elaborate.
Previously, it was only known that the expanding search problem is~$\mathsf{NP}$-hard~\cite{averbakh2012flowtime}.

\subsection{Further Related Work}
The unweighted pathwise search problem where all vertices have unit weight is also known as the \emph{traveling repairperson problem}.
Sahni and Gonzales~\cite{sahni1976complete} showed that the problem cannot be efficiently approximated within a constant factor unless~$\mathsf{P}=\mathsf{NP}$ on complete non-metric graphs when the searcher is required to take a Hamiltonian tour.
Afrati et al.~\cite{afrati1986complexity} considered the problem in metric spaces and gave an exact algorithm with quadratic runtime when the metric is induced by a path. This can be improved to linear runtime as shown by Garc\'ia et al.~\cite{garcia2002note}.
Minieka~\cite{minieka1989delivery} proposed an exact polynomial-time algorithm for the case that the metric is induced by an unweighted tree.
Sitters~\cite{sitters2002minimum} showed that the problem is~$\mathsf{NP}$-hard when the metric is induced by a tree with edge lengths~$0$ and~$1$.

The first approximation algorithm of the metric traveling repairperson problem is due to Blum et al.~\cite{blum1994minimum}, who gave a~$144$-approximation.
After a series of improvements~\cite{archer2010improved,archer2008faster,AroraK03,goemans1998improved,koutsoupias1996searching}, the currently best factor is a~$3.59$-approximation for general metrics~\cite{chaudhuri2003paths}, and a polynomial-time approximation scheme (PTAS) for trees~\cite{sitters2021polynomial} and Euclidean graphs~\cite{sitters2021polynomial}. 
Further variants of the problem have been studied both in terms of exact solution methods and in terms of competitive algorithms, e.g., variants with directed edges~\cite{fischetti1993delivery,friggstad2013asymmetric,nagarajan2008directed}, variants with processing times and time windows~\cite{tsitsiklis1992special}, variants with profits at vertices~\cite{dewilde2013heuristics}, variants with multiple searchers~\cite{chekuri2004maximum,fakcharoenphol2007traveling,li2018multiple}, and online variants~\cite{krumke2003news}.
The vertex-weighted version of the problem is often referred to as \emph{the} pathwise search problem. 
It has been shown to be~$\mathsf{NP}$-hard in metric graphs by Trummel and Weisinger~\cite{trummel1986complexity} and was further studied by Koutsoupias et al.~\cite{koutsoupias1996searching}. The approximation schemes by Sitters~\cite{sitters2021polynomial} also apply to the weighted case.

The expanding search problem has received considerably less attention in the literature than the pathwise problem. It has been shown to be~$\mathsf{NP}$-hard by Averbakh and Pereira~\cite{averbakh2012flowtime}. Alpern and Lidbetter~\cite{alpern2013mining} introduced an exact polynomial-time algorithm for the case when the graph is a tree and gave heuristics for general graphs. 
Averbakh et al.~\cite{averbakh2012emergency} considered a generalization of the problem with multiple searchers when the underlying graph is a path; Tan et al.~\cite{tan2019scheduling} considered multiple searchers in a tree network.
Hermans et al.~\cite{HermansLM22} devised an~$8$-approximation algorithm that is based on the exact algorithm for trees~\cite{alpern2013mining}.
Angelopoulos et al.~\cite{angelopoulos2019expanding} studied the expanding search ratio of a graph defined as the minimum over all expanding searches of the maximum ratio of the time to reach a vertex by the search algorithm and the time to reach the same vertex by a shortest path.
They showed that this ratio is~$\mathsf{NP}$-hard to compute and gave a search strategy that achieves a~$(4\ln 4)$-approximation of the optimum. 

The pathwise and expanding search problems also appear naturally as strategies for the seeker in a two-player zero-sum game between a hider and a seeker, where the hider chooses a vertex that maximizes the expected search time. In contrast, the seeker aims to minimize the search time. 
Gal~\cite{gal1979search} computed the value (i.e., the unique search time in an equilibrium of the game) of the pathwise search game on a tree.
Alpern and Lidbetter~\cite{alpern2013mining} computed this value for the expanding search game; see also~\cite{alpern2019approximate} for approximations of this value for general graphs. For more details on search games, we refer to~\cite{alpern2013search,alpern2006theory,gal1980search,isaacs1965differential,kirkpatrick2009hyperbolic}.

\section{Preliminaries}
\label{sec:prelim}
We define~$\N = \{0, 1, 2, \dots\}$ and use~$\N_{>0}$ to refer to~$\N \setminus \{0\}$.
For a natural number~$n \in \N$, we write~$[n] = \{1, 2, \dots, n\}$.

In the expanding search problem, we are given a connected, undirected graph~$G = (V, E)$ with~$|V| = n$ and a designated root vertex~$r \in V$.
Further, we are given for each vertex~$v \in V$ has a weight~$w_v \in \N$.
We denote by~$V^* \subseteq V$ the set of vertices with~$w_v > 0$.
For a subgraph~$H$ of~$G$ we define~$w(H) = \sum_{v \in V(H)} w(v)$ to be the sum of weights of vertices in~$H$.
Finally, for each edge~$e \in E$ we are given a length~$\length_e \in \N$.

We consider an agent initially located at the root~$r$, who performs an \emph{expanding search pattern}~$\sigma = (e_1, \dots, e_m)$, where~$m \leq n - 1$, satisfying the following properties:
\begin{enumerate}[leftmargin=*]
    \item The root~$r$ is incident to~$e_1$.
    \item For all~$i \in \{1, \dots, m\}$, the set~$\{e_1, \dots, e_i\}$ forms a tree in~$G$.
\end{enumerate}

For a vertex~$v \in V^* \setminus \{r\}$, let
$$
k_v(\sigma) = \inf \bigl\{i \in [m] : v \in e_i \bigr\}
$$
be the index of the first edge in~$\sigma$ that reaches~$v$, and set~$k_r(\sigma) = 0$.
We then define the \emph{latency} of a vertex~$v \in V^*$ under the expanding search pattern~$\sigma$ as
$$
\Length_v(\sigma) = \sum_{i=1}^{k_v(\sigma)} \length_{e_i}\ .
$$
Given the graph~$G$ with vertex weights and edge lengths, in the expanding search problem, the goal is to find an expanding search pattern~$\sigma$ that minimizes the \emph{\obj}%
, defined as
$$
\Length(\sigma) = \sum_{v \in V^*} w_v \Length_v(\sigma)\ .
$$
Note that vertices~$v$ with~$w_v = 0$ do not contribute to the objective function and, therefore, do not need to be visited.
However, they may still be used as Steiner vertices in the constructed search trees and cannot be contracted, as is possible in the pathwise search problem.
When the pattern~$\sigma$ is clear from the context, we omit the explicit dependence on~$\sigma$ and simply write~$\Length$,~$\Length_v$, and~$k_v$.
The \emph{length}~$\ell (\sigma)$ of a search pattern~$\sigma$ is defined as the total sum of the edge lengths
$$
\ell (\sigma ) = \sum_{e\in \sigma}\ell_e \ .
$$
Finally, for two expanding search patterns~$\sigma=(e_1,\dots,e_m),\sigma'=(e'_1,\dots,e'_{m'})$, we define their concatenation~$\sigma + \sigma'$ as the subsequence of~$(e_1, \dots, e_m, e'_1, \dots, e'_{m'})$ where any edge that closes a cycle with previous edges in the sequence is excluded.

Further, we also define Euclidean graphs:
A~$d$-dimensional Euclidean graph is a complete graph where vertices represent points in a~$d$-dimensional Euclidean space, and the edge lengths are the Euclidean distances between the vertices. For Euclidean instances, we assume that instead of the edge lengths, the coordinates of the vertices are given.

\section{The Unweighted Case}
\label{app:unweighted}

We first give an approximation algorithm for the special case where all vertex weights are equal to~$1$, i.e., we show the following result.

\begin{theorem}\label{thm:unweighted}
There is a polynomial-time~$2\mathrm{e}$-approximation algorithm for the unweighted expanding search problem.
\end{theorem}

Our algorithm is an adaptation of the approximation algorithm by Chaudhuri et al.~\cite{chaudhuri2003paths} for the pathwise search problem, where the objective is to find a path in an undirected graph with edge lengths that minimizes the sum of latencies of the vertices.

Like the approximation algorithm of Chaudhuri et al., our approach is based on approximate solutions to several~$k$-minimum spanning tree ($k$-MST) instances.
However, the way in which these approximate solutions are combined to form an approximate solution to the original problem differs.
In the rooted version of the~$k$-MST problem, we are given an unweighted connected graph~$G = (V,E)$ with a designated root vertex~$r \in V$ and non-negative edge lengths~$\length_e \in \mathbb{N}$,~$e \in E$. Let~$\mathcal{T}_k$ denote the set of all trees~$T = (V_T, E_T)$ that are subgraphs of~$G$ with~$|V_T| = k$ and~$r \in V_T$. The~$k$-MST problem is concerned with finding a tree~$T \in \mathcal{T}_k$ that minimizes~$\length(T) = \sum_{e \in E_T} \length_e$. This problem is~$\mathsf{NP}$-complete, as shown by Ravi et al.~\cite{RaviSMRR94}.

For all~$k \in [n]$, we solve this problem approximately using the~$2$-approximation algorithm by Garg~\cite{Garg05} and obtain~$n$ trees~$T_1, T_2, \dots, T_n$, where~$T_1 = (\{r\}, \emptyset)$ is the empty tree consisting only of the root vertex.
We proceed to construct an auxiliary weighted directed graph~$H = (V_H, A_H)$ with the vertex set~$V_H = [n]$ and the arc set~$A_H = \{(i, j) \in V_H^2 : i < j\}$.
The directed graph~$H$ is constructed such that any~$(1, n)$-path~$P$ corresponds to some expanding search pattern~$\sigma_P$ such that we can upper-bound~$\Length(\sigma_P)$ by~$c(P)$.
For this, a vertex~$j \in V_H$ models the exploration of tree~$T_j$. The cost~$\cost_{i,j}$ of an edge from~$i \in V_H$ to~$j \in V_H$ models the delay in the exploration of the vertices not explored in~$T_i$ due to the exploration of~$T_j$.
This way, we obtain an upper bound on the latency of the vertices by assuming that these vertices will only be explored after~$T_j$ has been fully explored.
Traversing the edge~$(i,j)$, the exploration of~$n-i$ vertices is delayed by~$\length(T_j)$; we therefore set~$\cost_{i,j} = (n-i)\length(T_j)$.
To obtain an approximate expanding search pattern, we compute a shortest~$(1,n)$-path~$P = (n_0, n_1, \dots, n_l)$ in~$H$ with~$n_0 = 1$,~$n_l = n$, and some~$l \in \mathbb{N}$.
Hence, the expanding search pattern consists of~$l$ phases.
In phase~$j \in [l]$, we explore all edges~$e \in E[T_{n_j}]$ with~$\smash{|e \cap (\bigcup_{i=0}^{j-1} V[T_{n_i}])| < 2}$ in an order such that the subgraph of explored vertices is always connected.
Here, we denote by~$E[T_{n_j}]$ and~$V[T_{n_j}]$ the edge and vertex set of tree~$T_{n_j}$, respectively.
In this way, when phase~$j$ is finished, all vertices in~$V[T_{n_j}]$ have been explored, and the total length of edges used in phase~$j$ is at most~$\length(T_{n_j})$.
Since~$n_l = n$, all vertices have been explored when the algorithm terminates.
Formally, the algorithm is given as follows:
\begin{enumerate}[label=\arabic*),leftmargin=*]
    \item For all~$k\in [n]$, solve the~$k$-MST problem with the 2-approximation algorithm of Garg~\cite{Garg05} and obtain~$n$ trees~$T_1,\dots,T_{n}$.
    \item Construct the auxiliary weighted directed graph~$H=(V_H,A_H)$ with the vertex set~$V_H = \{1,2,\dots,n\}$, the arc set~$A_H=\{(i,j) \in V_H^2 : i<j\}$, and the arc~costs~$\cost_ {i,j}  = (n-i)\length(T_j)$.
    \item Compute a shortest~$(1,n)$-path~$P = (n_0,n_1,\dots,n_l)$ with~$n_0 = 1$ and~$n_l = n$ in~$H$.
    \item For each phase~$j\in [l]$ explore all unexplored vertices of~$V[T_{n_j}]$ in any feasible order using the edge set of~$E[T_{n_j}]$. 
\end{enumerate}

Let~$\sigma_{\alg}$ be the expanding search pattern obtained from this algorithm.
Let~$v_i$ be the~\mbox{$i$-th} vertex explored according to~$\sigma_{\alg}$  and let~$j(i) \in [l]$ be chosen such that~$n_{j(i)-1} < i \leq n_{j(i)}$.
We define~$\smash{\pi(i) = \sum_{k=0}^{j(i)} \length(T_{n_k})}$.
The following lemma gives an upper bound on each vertex's latency.

\begin{lemma}
\label{lem:pi-upper-bound-vertex}
The latency of~$v_i$ in~$\sigma_\alg$ is bounded by   ~$\Length_{v_i}(\sigma_{\alg})\leq \pi(i)$.
\end{lemma}
\begin{proof}
    Since~$n_{j(i)-1} < i \leq n_{j(i)}$, vertex~$v_i$ will be explored in or before phase~$j(i)$. To give an upper bound on the latency~$\Length_{v_i}(\sigma_{\alg})$ of~$v_i$ in the expanding search pattern~$\sigma_{\alg}$, note that~$v_i$ is explored at the latest if all trees~$T_{n_k}$,~$1 \leq k \leq j(i)$, are nested by inclusion and~$v_i$ is visited at the very end of the exploration of~$T_{n_{j(i)}}$.
    Also, note that the total length of added edges in any phase~$k$ is at most~$\ell(T_{n_{k}})$.
    Thus, the latency of vertex~$v_i$ in~$\sigma_{\alg}$ can be bounded from above by 
    \begin{align*}
        \Length_{v_i}(\sigma_{\alg})\leq \sum_{k=0}^{j(i)} \length(T_{n_k}) =\pi(i) \ .
    \end{align*}
    This completes the proof.
\end{proof}

The following lemma bounds the \obj~$\Length(\sigma_{\alg})$. Similar lemmas have been proven in related settings by Goemans and Kleinberg~\cite{goemans1998improved} and Chaudhuri et al.~\cite{chaudhuri2003paths}.

\begin{lemma}\label{lem:pathinH}
For the \obj\ of the algorithm, we have~$\Length(\sigma_{\alg}) \leq z$ where~$z$ is the cost of a shortest~$(1,n)$-path in~$H$. 
\end{lemma}

\begin{proof}
Let~$P = (n_0,\dots,n_l)$ be a shortest~$(1,n)$-path in~$H$. Its cost is equal to
\begin{align*}
\cost(P) =\sum_{j=1}^l (n-n_{j-1}) \length(T_{n_j}) \ .
\end{align*}
Next, consider~$\sigma_{\alg}$ and recall from \Cref{lem:pi-upper-bound-vertex} that we can bound the latency of the~$i$-th vertex~$v_i$ in~$\sigma_\alg$ by~$\pi(i)$.
Taking the sum over all vertices, we obtain
\begin{multline*}
\Length(\sigma_{\alg})
\leq \sum_{i=1}^{n} \pi(i)
=\sum_{i=1}^{n} \sum_{k=0}^{j(i)} \length(T_{n_k})
= \sum_{j=1}^{l} (n_j - n_{j-1}) \sum_{k=0}^j \length(T_{n_k})\\
=\sum_{j=1}^l (n-n_{j-1}) \length(T_{n_j})
=\cost(P) \ .
\end{multline*}
Thus, the~\obj\ is bounded from above by the cost of path~$P=(n_0,n_1,\dots,n_l)$ in~$H$, as claimed.
\end{proof}

We now claim that the cost of the shortest~$(1,n)$-path~$P$ in~$H$ is at most~$2\mathrm{e}$ times the \obj\ along the optimal sequence.
To prove this claim, we consider a randomized sequence of exponentially growing subtrees of~$G$ and show that the path along their corresponding vertices in~$H$ has the desired property.

\begin{lemma}\label{lem:bound2e}
Let~$\sigma^*$ be an optimal expanding search pattern with a total latency of~$\Length^* = \Length(\sigma^*)$. Then, a shortest~$(1,n)$-path in~$H$ has cost at most~$2\mathrm{e}\Length^*$.
\end{lemma}

\begin{proof}
Our goal is to construct a~$(1,n)$-path in~$H$ and compare its cost to the \obj\ of the optimal expanding search pattern.
To do so, let~$\Length^*_i$ denote the latency of the~$i$-th explored vertex in the optimal expanding search pattern.
Observe that no expanding search pattern can explore the~$i$-th vertex with latency less than the length of an optimal~$i$-MST.
Hence, this is also true for the optimal expanding search pattern. Since we use the 2-approximation algorithm by Garg~\cite{Garg05} to solve the~$k$-MST problem, we obtain that~$\length(T_i)\leq 2 \Length^*_i$.
To show the result, let~$\gamma > 1$ and~$b\in [1,\gamma)$ be two parameters to be optimized later.
Let~$\omega \in \mathbb{Z}$ be the smallest number such that~$\gamma^\omega > n \length(T_{n})$,
and let~$\alpha \in \mathbb{Z}$ be the largest number such that~$2b\gamma^\alpha<\min_{e \in E} \length_e$.
Then, for all~$j \in \{\alpha,\dots,\omega\}$, let
\begin{align*}
n_j = \max \bigl\{k \in [n] : \length(T_{k}) \leq 2b\gamma^j \bigr\} \ ,
\end{align*}
i.e.,~$n_j$ is defined to be the largest number of vertices that can be visited by one of the~$i$-MSTs~$T_1,T_2,\dots,T_{n}$ computed with the 2-approximation algorithm of Garg~\cite{Garg05} such that the length of the tree is bounded from above by~$2b\gamma^j$.
Note that these values are well-defined since~$\length(T_1)=0$ and that by the choices of~$\alpha, \omega \in \mathbb{Z}$, we have~$n_{\alpha} = 1$ and~$n_{\omega} = n$, since~$\ell(T_n) < \frac{1}{n} \gamma^\omega \leq 2b\gamma^\omega$.
Consider the sequence~$n_{\alpha}, n_{\alpha +1}, \dots, n_{\omega}$  and note that the sequence is non-decreasing.
We may assume that it is strictly increasing as we consider the inclusion-wise largest strictly increasing subsequence of~$n_{\alpha}, n_{\alpha +1}, \dots, n_{\omega}$ otherwise.
Thus, the sequence corresponds to a~$(1,n)$-path given by~$P = (n_\alpha,n_{\alpha+1},\dots,n_{\omega})$ in~$H$.
We proceed to compute its expected cost.

To this end, let~$b = \gamma^U$ where~$U$ is a random variable distributed uniformly in~$[0,1)$.
The parameter~$\gamma$ will be determined later.
Let~$\sigma^*$ be an optimal expanding search pattern and let~$v_1^*,\dots,v_{n}^*$ with~$r = v_1^*$ be the vertices as they are explored by~$\sigma^*$.
Further let~$i \in \{2,3,\dots,n\}$ be arbitrary and let~$j \in \{\alpha,\dots,\omega\}$ and~$d \in [1,\gamma)$ be such that
$\Length_{i}^*=  d\gamma^{j}$.
We brief{}ly argue that such values always exist.
First, note that with~$i\geq 2$, we have~$L_i^*\geq \min_{e \in E} \length_e > 2b\gamma^\alpha>\gamma^\alpha$.
Furthermore,~$L_i^*$ can be bounded by~$L_i^* \leq \sum_{j=1}^i \ell(T_j^*) \leq n  \ell(T_n)<\gamma^{\omega}$, where~$T_j^*$ is the optimal~$j$-MST.
We distinguish two cases:

\subsection*{First Case:~$d\leq b$}
Since there exists a tree containing the root with length at most~$d\gamma^j \leq b\gamma^j$ that contains at least~$i$ vertices, we know that
our~$2$-approximation of~$i$-MST has length at most~$2b \gamma^j$.
Thus,~$n_j \geq i$ holds.
This allows us to bound~$\pi(i)$ from above by
\begin{align*}
    \pi(i)
    \leq \sum_{k=\alpha}^{j} \length(T_{n_k}) \leq \sum_{k=\alpha}^j 2b\gamma^k 
    < \sum_{k=-\infty}^j 2b\gamma^k
    =2b\frac{\gamma^{j+1}}{\gamma-1}
    =2b\gamma^j\frac{\gamma}{\gamma-1} \ .
\end{align*}

\subsection*{Second Case:~$d > b$}
We have that~$d \gamma^j < \gamma^{j+1}$ since~$d <\gamma$.
This implies that there is a tree containing the root with length at most~$\gamma^{j+1}$ containing at least~$i$ vertices. Analogously to the argumentation in the first case, we obtain~$n_{j+1} \geq i$. This allows us to bound~$\pi(i)$ from above by
\begin{align*}
    \pi(i)
\leq \sum_{k=\alpha}^{j+1}\length(T_{n_k}) \leq \sum_{k=\alpha}^{j+1} 2b\gamma^k < \sum_{k=-\infty}^{j+1} 2b\gamma^k
= 2b\frac{\gamma^{j+2}}{\gamma-1} 
= 2b\gamma^{j+1}\frac{\gamma}{\gamma-1} \ .
\end{align*} 

We have~$U\in [\log_\gamma d,1]$ in the first case, and~$U\in [0,\log_\gamma d)$ in the second.
Taking the expectation over~$U$, we obtain
\begin{align*}
    \mathbb{E}_U\bigl[\pi(i)\bigr]
    &\leq \int_{\log_\gamma d}^1 2b\gamma^j\frac{\gamma}{\gamma-1} \,\mathrm{d}U + \int_0^{\log_\gamma d} 2b\gamma^{j+1}\frac{\gamma}{\gamma-1} \,\mathrm{d}U\\
    &=2\gamma^j\frac{\gamma}{\gamma-1}\left[ \int_{\log_\gamma d}^1 \gamma^U \,\mathrm{d}U + \gamma\int_0^{\log_\gamma d} \gamma^U \,\mathrm{d}U \right]\\
    &=2\gamma^j\frac{\gamma}{\gamma-1}\left[ \frac{\gamma-d}{\ln \gamma} + \gamma\frac{d-1}{\ln \gamma} \right]\\
    &=2\gamma^jd\frac{\gamma}{\ln \gamma}\\
    &= 2\Length_i^* \frac{\gamma}{\ln \gamma} \ .
\end{align*}
Therefore, using~$\cost(P)=\sum_{i=1}^{n} \pi(i)$, we obtain
\begin{align*}
\frac{\mathbb{E}_U[\cost(P)]}{\Length^*}
=\frac{\mathbb{E}_U[\sum_{i=1}^n \pi(i)]}{\Length^*}
= \frac{\sum_{i=1}^n \mathbb{E}_U\bigl[\pi(i)\bigr]}{\sum_{i=1}^n \Length_i^*}
\leq \frac{2 \frac{\gamma}{\ln \gamma}\sum_{i=1}^n \Length_i^*}{\sum_{i=1}^n \Length_i^*}
=2\frac{\gamma}{\ln \gamma} \ .
\end{align*}
This quantity is minimized for~$\gamma=\mathrm{e}$ for which the approximation ratio turns out to be~$2\mathrm{e}\approx 5.44$.
Hence, the randomized~$(1,n)$-path~$P$ has expected cost at most~$2\mathrm{e}\Length^*$.
Therefore, the cost of a shortest~$(1,n)$-path in~$H$ is also bounded by~$2\mathrm{e}\Length^*$.
\end{proof}

Together, Lemma~\ref{lem:pathinH} and Lemma~\ref{lem:bound2e} imply \Cref{thm:unweighted}.

\section{The Weighted Case}
\label{sec:weighted}

In this section, we prove our main result for the weighted setting, i.e., we prove the following theorem.

\begin{theorem}
    \label{thm:weighted}
     There is a polynomial-time~$(2\mathrm{e}+\eps)$-approximation algorithm for the expanding search problem with weights~$w_v \in \N$ for every~$\eps>0$. 
\end{theorem}

The proof consists of two parts.
First, we provide a polynomial-time~$2\mathrm{e}$-ap\-prox\-i\-mation algorithm for binary weights, i.e., the case in which each vertex has a weight of either~$0$ or~$1$.
In the second step, we reduce the problem with general weights to the binary case. 
In this reduction, we lose a factor of~$(1+ \varepsilon)$ in the approximation guarantee. 

The following two sections are devoted to these two results.

\subsection{A~$2\mathrm{e}$-Approximation Algorithm for Binary Weights}

In this section, we provide a polynomial-time~$2\mathrm{e}$-approximation algorithm for binary weights, i.e., we show the following result.
The algorithm uses the same algorithmic ideas as the algorithm for unit weights, but gives sufficient priority to the vertices with weight~$1$.

\begin{lemma}
\label{lem:0-1-weights}
 There is a polynomial-time~$2\mathrm{e}$-approximation algorithm for the expanding search problem with binary weights~$w_v \in \{0,1\}$. 
\end{lemma}

\begin{proof}
Let~$G=(V, E)$ be an instance of the expanding search problem where every vertex has a weight of either~$0$ or~$1$.
Let~$V^0\subset V$ be the vertices in~$V$ with weight~$0$, and analogously, let~$V^1\subset V$ be the vertices in~$V$ with weight~$1$.
We define~$n^0 = |V^0|$ and~$n^1 = |V^1|$ such that~$n = n^0+n^1=|V|$. We may assume that~$n^1 > 0$ and~$n^0 > 0$ since the result is trivial or follows from \Cref{thm:unweighted} otherwise. 

We construct an instance~$G'=(V',E')$ where all vertices have a weight of~$1$ as follows:
We start with~$G=(V,E)$ and set all vertex weights to~$1$.
Next, for each vertex~$v\in V'$ corresponding to a vertex in~$V^1$, we add~$2n-1$ copies of~$v$ and connect these to~$v$ by an edge of length~$0$.
Thus, the vertices in~$V'$ can be categorized in \emph{original} vertices, \emph{copies}, and vertices corresponding to a vertex of weight~$0$ in~$V$.
This finishes the construction and we obtain~$w(G')=2nn^1+n^0=2nw(G)+n^0<(2w(G)+1)n$ since~$n^0<n$.
Further, note that for any tree~$T$ in~$G$ there is a corresponding tree~$T'$ in~$G'$ with~$2nw(T)\leq w(T')$ and~$c(T)=c(T')$.

Next, we run an algorithm similar to the one given in \Cref{app:unweighted}:
\begin{enumerate}[label=\arabic*),leftmargin=*]
    \item For all~$k\in \{2n,4n,\dots,2n^1n\}$, solve the~$k$-MST problem on~$G'$ with the 2-approximation algorithm of Garg~\cite{Garg05} and obtain~$n^1$ trees~$T_1,\dots,T_{n^1}$.
    \item Construct the auxiliary weighted directed graph~$H =(V_H,A_H)$ with the vertex set~$V_H = \{1,2,\dots,n^1\}$, the arc set~$A_H=\{(i,j) \in V_H^2 : i<j\}$, and the arc costs~$\cost_ {i,j} =(n^1-i)\length(T_j)$.
    \item Compute a shortest~$(1,n^1)$-path~$P = (n_0,n_1,\dots,n_l)$ with~$n_0 = 1$ and~$n_l = n^1$ in~$H$.
    \item For each phase~$j\in [l]$ explore all unexplored vertices of~$V[T_{n_j}]$ in any feasible order using the edge set of~$E[T_{n_j}]$. 
\end{enumerate}
Let~$\sigma'$ be the obtained expanding search sequence for instance~$G'$.
The sequence does not necessarily have a finite objective value for instance~$G'$ as~$2n^1n < w(G')$.
However, starting from~$\sigma'$, we can still construct an expanding search sequence~$\sigma$ for instance~$G$ with finite objective value.
First, note that since~$2n^1n - n^0 > 2n(n^1 -1)$ it follows that every original vertex~$v\in V'$ is visited by~$\sigma'$.
Furthermore, as all copies in~$V'$ are connected to an original vertex by an edge of length~$0$, we may adjust~$\sigma'$ such that the first visit of an original vertex~$v\in V'$ is immediately followed by a visit of all of its copies.
Thus, we obtain a new expanding search sequence~$\sigma''$ visiting every vertex visited by~$\sigma'$ without increasing any latency.
Moreover,~$\sigma''$ visits all original vertices and all copies in~$V'$.
To obtain the expanding search sequence~$\sigma$ for instance~$G$ from~$\sigma''$, we skip all edges in~$\sigma''$ that correspond to edges connecting an original vertex to one of its copies.
Since~$\sigma''$ visits all original vertices, the sequence~$\sigma$ has a finite objective value.

It remains to show that sequence~$\sigma$ yields the desired approximation ratio of~$2\mathrm{e}$.
To this end, observe that vertex~$v\in V^1$ has the same latency in~$\sigma$ as its corresponding original vertex~$\smash{v'\in V'}$ and all its copies in the sequence~$\sigma''$.
For~$i \in [n^1]$ let~$j(i) \in  [l]$ be such that~$\smash{n_{j(i) - 1} < 2ni\leq n_{j(i)}}$.
Because~$\smash{n^0 < n}$, tree~$\smash{T_{j(i)}}$ contains at least~$i$ original vertices.
Then we define~$\smash{\pi(i)= \sum_{k=0}^{j(i)}\ell(T_{n_k})}$.
Let~$v_i$ be the~$i$-th vertex of~$V^1$ visited by~$\sigma$ and let~$v_i'$ be the corresponding original vertex in~$\smash{V'}$.
With an analogous computation as in \Cref{lem:pi-upper-bound-vertex} and \Cref{lem:pathinH} we obtain
\begin{multline*}
 \Length(\sigma)
=  \sum_{i=1}^{n^1}L_{v_i}(\sigma)
=  \sum_{i=1}^{n^1}L_{v_i'}(\sigma'')
\leq  \sum_{i=1}^{n^1} \pi(i)
=  \sum_{i=1}^{n^1} \sum_{k=0}^{j(i)} \length(T_{n_k})\\
= \sum_{j=1}^{l} (n_j - n_{j-1}) \sum_{k=0}^j \length(T_{n_k})
=\sum_{j=1}^l (n^1-n_{j-1}) \length(T_{n_j})
=\cost(P) \ ,
\end{multline*}
where~$P$ is a shortest~$(1,n^1)$-path in~$H$.

Finally, we show that the cost of a shortest~$(1,n^1)$-path in~$H$ is upper bounded by~$2\mathrm{e}$ times the total latency of an optimal expanding search pattern for instance~$G$.
In particular, we use the same technique as in the proof of \Cref{lem:bound2e}. 
We let~$\sigma^*$ be an optimal expanding search sequence for~$G$ and denote by~$L_i^*$ the latency of the~$i$-th vertex of~$V^1$ visited by~$\sigma^*$. 
Using the 2-approximation algorithm by Garg~\cite{Garg05}, we again obtain that~$\length(T_i)\leq 2 \Length^*_i$.
The definitions of~$n_j, \alpha, \omega, \gamma,$ and~$b$ and the construction of the randomized~$(1,n^1)$-path~$P$ in~$H$  are the same as in the proof of \Cref{lem:bound2e}.
Following its line of argumentation, we obtain that 
\begin{align*}
    \mathbb{E}_U\bigl[\pi(i)\bigr]
    &= 2\Length_i^* \frac{\gamma}{\ln \gamma}
\end{align*}
for all~$i\in [n^1]$.
In total, this yields that 
\begin{align*}
\frac{\mathbb{E}_U[\cost(P)]}{\Length^*}
=\frac{\mathbb{E}_U[\sum_{i=1}^{n^1} \pi(i)]}{\Length^*}
= \frac{\sum_{i=1}^{n^1} \mathbb{E}_U\bigl[\pi(i)\bigr]}{\sum_{i=1}^{n^1} \Length_i^*}
\leq \frac{2 \frac{\gamma}{\ln \gamma}\sum_{i=1}^{n^1} \Length_i^*}{\sum_{i=1}^{n^1} \Length_i^*}
=\frac{2\gamma}{\ln \gamma} \ .
\end{align*}
Setting~$\gamma=\mathrm{e}$ implies
\begin{align*}
    \mathbb{E}_U[\cost(P)] \leq 2\mathrm{e}\Length^* \ .
\end{align*}
Hence, choosing a shortest~$(1,n)$-path~$P$ in~$H$ yields 
\begin{align*}
    L(\sigma)\leq c(P) \leq 2\mathrm{e}\Length^*
\end{align*}
as claimed.
\end{proof}

\subsection{Reducing the Weighted Case to Binary Weights}
In this section, we prove the following result.

\begin{lemma}
\label{lem:reduction:weighted-01}
    Given a polynomial-time~$\alpha$-approximation algorithm for the expanding search problem with binary weights, there is a polynomial-time~$(1 + \varepsilon) \alpha$-approx\-i\-mation algorithm for the expanding search problem on weighted graphs.
\end{lemma}

Our approach has two steps inspired by Sitters~\cite{sitters2021polynomial}. 
In the first step (\Cref{lem:euclidean-red1}), we show a reduction from the expanding search problem to the so-called the~$\delta$-bounded expanding search problem, for some constant~$\delta\in\mathbb{R}_+$, in the sense that a polynomial-time~$\alpha$-approximation for the weighted~$\delta$-bounded expanding search problem implies a polynomial-time~$(\alpha + \varepsilon)$-approximation for the weighted expanding search problem. 
In the second step (\Cref{lem:rounding}), we show that if there is a polynomial-time~$\alpha$-approximation algorithm for the weighted~$\delta$-bounded expanding search problem with weights in~$\{0,1\}$, then there is a polynomial-time~$(\alpha + \varepsilon)$-approximation algorithm for the~$\delta$-bounded expanding search problem.
Hence, both results combined imply \Cref{lem:reduction:weighted-01}.

In the~$\delta$-bounded expanding search problem, the input is as for the expanding search problem, but comes with an additional delay parameter~$D\geq 0$.
Further, there is a parameter~$\delta$, not part of the input, such that there exists a solution that visits all vertices with non-zero weight and has length at most~$\delta D$.
The objective is to minimize~$\smash{\Length^{D}(\sigma)=\sum_{v \in V^*} w_v \Length^{D}_v(\sigma)}$ where~$\Length^D_v(\sigma)=D+\Length_v(\sigma)$. 

In the following, we assume~$0 < \varepsilon \leq 1$ and use~$O_\eps(f)$ to denote~$O(f)$ when~$\varepsilon$ is a constant.
\medskip

\textbf{Reducing the Expanding Search Problem to the~$\delta$-Bounded Expanding Search Problem.}

We show the following lemma.
\begin{lemma}\label{lem:euclidean-red1}
    Consider any class~$\mathcal{C}$ of graphs with edge lengths and constants~$\alpha>1$ and~$\eps\in(0,1]$. There exists a constant~$\delta>0$ such that, if there exists a polynomial-time~$\alpha$-approximation algorithm for the~$\delta$-bounded expanding search problem on~$\mathcal{C}$, then there exists a po\-ly\-no\-mial-time~$(\alpha+\eps)$-approximation algorithm for the expanding search problem on~$\mathcal{C}$. 
\end{lemma}

We follow the decomposition approach by Sitters~\cite{sitters2021polynomial} and adapt it to the expanding search problem at several places.
To do so, we assume that a polynomial-time~$\alpha$-approximation algorithm for the~$\delta$-bounded expanding search problem on~$\mathcal{C}$ for a yet-to-be-determined value of~$\delta$ is given and we denote it by~$\textsc{Approx}_{\text{$\delta$-bd}}$.
In the remainder of this subsection, we first construct a polynomial-time algorithm for the expanding search problem on~$\mathcal{C}$ based on~$\textsc{Approx}_{\text{$\delta$-bd}}$ and a given~$\eps>0$.
Afterward, we show that this algorithm yields an approximation guarantee of~$(\alpha+O(\eps))$.

In that direction, we further require a polynomial-time~$\beta$-approxi\-mation algorithm~$\textsc{Approx}_\beta$ for the expanding search problem on~$\mathcal{C}$ for some constant~$\beta$.
We emphasize that \emph{any} constant~$\beta$ is sufficient to obtain an approximation guarantee of~$\alpha+\varepsilon$ in polynomial time. Therefore, we can pick, e.g., the algorithm from~\cite{HermansLM22}. 
For notational purposes, we assume~$\alpha\neq\beta$.

Our algorithm is now designed as follows.
First, we apply~$\textsc{Approx}_\beta$.
This yields an order of the vertices according to their search times in the solution.
Next, we partition the vertices by cutting this order at several places and run~$\textsc{Approx}_{\text{$\delta$-bd}}$ on the (carefully defined) emerging subinstances.
Note that despite their low total latencies, these subsolutions may have large total length.
Thus, we cannot simply concatenate them to obtain a solution for the original instance as this would delay the subsolutions of all later subinstances.
We solve this issue by cutting each subsolution at a certain point and using the subsolution given by~$\textsc{Approx}_\beta$ from then on---a solution with a length bound.

In the following, we present our algorithm for the expanding search problem in more detail.
To this end, let~$I$ be an instance of the expanding search problem on a metric space~$\mathcal{C}$ and let~$\varepsilon>0$ be given.
The algorithm consists of five steps.

\medskip
\begin{enumerate}[label=\arabic*),leftmargin=*]
    \item \textbf{Approximate:} Apply~$\textsc{Approx}_\beta$ to instance~$I$ and obtain solution~$\sigma_\beta$.
    
    \item \textbf{Partition:}
    \begin{itemize}
        \item Define~$\gamma = \frac 3\varepsilon$,~$a = \frac{\beta \gamma}{\varepsilon}$, and draw~$b$ uniformly at random from~$[0, a]$.
        \item Let~$q\in \N$ be the smallest number such that~$\Length_v(\sigma_\beta) < \mathrm{e}^{(q-1) a+b}$ for all~$v\in V$ and define time points~$t_i = \mathrm{e}^{(i-2) a+b}$ for~$i\in[q+1]$.
        \item For~$i\in[q]$, let~$V_i = \{ v\in V : t_i \leq \Length_v(\sigma_\beta) < t_{i+1} \}$ and~$U_i = V_1 \cup \dots \cup V_i$. 
        \item For~$i\in[q]$, define~$I_i$ to be the instance obtained from~$I$ by setting the weight of all vertices in~$V\setminus V_i$ to~$0$.
        Note that~$I_i$ with delay parameter~$D_i = \gamma t_i$ is an instance of the~$\delta$-bounded expanding search problem with~$\delta = \frac{\mathrm{e}^a}{\gamma}$ since
        \begin{align*}
        \delta D_i = \frac{\mathrm{e}^a}{\gamma} \gamma t_i = \mathrm{e}^{a} t_i  = t_{i+1}\ ,
        \end{align*}
        and by definition, all vertices in~$V_i$ can be explored with a sequence of length~$t_{i+1}$.
    \end{itemize}
    \item \textbf{Approximate Subproblems:} For~$i\in[q]$, apply~$\textsc{Approx}_{\text{$\delta$-bd}}$ to~$I_i$ and obtain an~$\alpha$-approximation~$\sigma_{\alpha,i}$ for the~$\delta$-bounded expanding search problem on~$I_i$ with~$\delta=\frac{\mathrm{e}^a}{\gamma}$.
    \item \textbf{Modify:} For each~$i\in[q]$, define~$\sigma_i$ to be~$\sigma_{\alpha,i}'+\sigma_{\beta,i+1}$ where:
    \begin{itemize}
        \item~$\sigma_{\alpha,i}'$ is the longest prefix of~$\sigma_{\alpha,i}$ of length at most~$(1+\frac{\mathrm{e}^a}{\eps\gamma}) \gamma t_i$ and
        \item~$\sigma_{\beta,i+1}$ is the prefix of~$\sigma_{\beta}$ visiting~$U_{i+1}$.
    \end{itemize}

    \item \textbf{Concatenate:} Return~$\sigma=\sigma_1+\dots+\sigma_q$.
\end{enumerate}
\medskip

We prove Lemma~\ref{lem:euclidean-red1} by establishing two intermediate results about the algorithm described above.
First, we demonstrate that dividing the instance~$I$ into subinstances~$I_i$ of the~$\smash{\frac{\mathrm{e}^a}{\gamma}}$-bounded expanding search problem incurs a loss of at most a factor of~$(1+\varepsilon)$ in the optimal objective-function value.
To formalize this, let~$\sigma^*$ represent an optimal solution for~$I$, and let~$\sigma^*_i$ denote an optimal solution for each subinstance~$I_i$ with~$i \in [q]$.

\begin{lemma}\label{lem:euclidean-red1-sub1}
    It holds that \smash{$\mathbb{E}\left[\sum_{i=1}^q \Length^{D_i}(\sigma^*_i)\right]\leq (1+\eps) \Length(\sigma^*).
$}
\end{lemma}

\begin{proof}
    First, observe that~$\sigma^*$ is a solution to~$I_i$, for any~$i\in[q]$. Therefore,
    \begin{align*}
        \sum_{v\in V_i^*} w_v (\gamma t_i+\Length_v(\sigma^*_i))\leq \sum_{v\in V_i^*} w_v (\gamma t_i+\Length_v(\sigma^*)) \ .
    \end{align*}
    Summing over all~$i\in [q]$ and taking expectations, we obtain
    \begin{align}
      \mathbb{E}\left[\sum_{i=1}^q\sum_{\smash{v\in V_i^*}} w_v (\gamma t_i+\Length_v(\sigma^*_i))\right]
        \leq &\;\mathbb{E}\left[\sum_{i=1}^q\sum_{\smash{v\in V_i^*}} w_v (\gamma t_i+\Length_v(\sigma^*))\right]\nonumber\\
        =&\;\sum_{\smash{v\in V^*}} w_v \Length_v(\sigma^*)+\sum_{\smash{v\in V^*}}\gamma w_v \mathbb{E}\left[t_{i(v)}\right] \ ,\label{eq:euclidean-red1-2}
    \end{align}
    where~$i(v)\in\N$ is defined such that~$v\in V_{i(v)}$ for each~$v\in V$.
    To analyze the second summand on the right-hand side, note that~$t_{i(v)}$ is a random variable of the form~$\smash{t_{i(v)} = \mathrm{e}^{-x} \Length_v(\sigma_\beta)}$ where~$x$ is uniform in~$[0, a]$.
    Hence,
\begin{equation*}
    \mathbb{E}\left[t_{i(v)}\right]
    = \frac{\Length_v(\sigma_\beta)}{a} \int_{0}^{a} \mathrm{e}^{-x} \,\mathrm{d}x = \frac{\Length_v(\sigma_\beta)}{a}(1- \mathrm{e}^{-a})  < \frac{\Length_v(\sigma_\beta)}{a} \ .
\end{equation*}
Together with Inequality~\eqref{eq:euclidean-red1-2}, this yields
\begin{align*}
    \mathbb{E}\left[\sum_{i=1}^q \Length^{D_i}(\sigma^*_i)\right]
    =&\; \mathbb{E}\left[\sum_{i=1}^q\sum_{\smash{v\in V_i^*}} w_v (\gamma t_i+\Length_v(\sigma^*_i))\right]\\
    \leq&\; \sum_{\smash{v\in V^*}} w_v \Length_v(\sigma^*)+\frac{\gamma}{a}\sum_{\smash{v\in V^*}}w_v \Length_v(\sigma_\beta)\\
    =&\;\sum_{\smash{v\in V^*}} w_v \Length_v(\sigma^*)+\frac{\eps}{\beta} \sum_{\smash{v\in V^*}}w_v \Length_v(\sigma_\beta)\\
    \leq&\; (1+\eps) \sum_{v\in V^*} w_v \Length_v(\sigma^*)\\
    =&\; (1+\eps) \Length(\sigma^*) \ ,
\end{align*}
using that~$\sigma_\beta$ is a~$\beta$-approximation in the last inequality.
\end{proof}

The following lemma addresses Step 4 of the algorithm. For each~$i \in [q]$, it provides an upper bound on the cost of~$\sigma_i$ relative to the cost of~$\sigma_i^*$ and establishes an upper bound on the total length of~$\sigma_i$.

\begin{lemma}\label{lem:euclidean-red1-sub2}
    For each~$i\in[q]$, the total length of~$\sigma_i$ is at most~$\gamma t_{i+1}-\gamma t_i$. Furthermore, it holds that
       ~$\Length^{D_i}(\sigma_i)\leq\alpha(1+\eps)  \Length^{D_i}(\sigma^*_i)$.
\end{lemma}

\begin{proof}
Consider an~$i\in[q]$. By construction, the total length of~$\sigma_i$ is at most 
\begin{multline*}
\left(1+\frac{\mathrm{e}^a}{\eps\gamma}\right) \gamma t_i+t_{i+1}
 =\left(1+\frac{\mathrm{e}^a}{\eps\gamma}+\frac{\mathrm{e}^a}{\gamma}\right) \gamma t_i
\leq \left(1 + \frac{\mathrm{e}^a}{3} + \frac{\mathrm{e}^a}{3} \right) \gamma t_i \\
= \left( \mathrm{e}^a - \frac{\mathrm{e}^a}{3} + 1 \right) \gamma t_i <(\mathrm{e}^a-1) \gamma t_i
=\gamma t_{i+1}-\gamma t_i \ ,
\end{multline*}
where we use~$\eps\leq 1$,~$\gamma\geq 3$, and~$a\geq 3$ for the inequalities. 

For the search times~$\Length_v(\sigma_i)$, we analyze each vertex~$v \in V_i$ with~$w_v > 0$ individually. First, observe that for any vertex~$v$ visited during the first part of~$\sigma_i$, namely~$\sigma_{\alpha,i}'$, the search time~$\Length_v(\sigma_i)$ remains the same as~$\Length_v(\sigma_{\alpha,i})$ and thus also~$\Length^{D_i}_v(\sigma_{i})=\Length^{D_i}_v(\sigma_{\alpha,i})$ holds.
For the remaining~$v\in V_i$, we may bound their latency by the total length of~$\sigma_i$, i.e.,~$\Length_v(\sigma_i)\leq (1+\frac{\mathrm{e}^a}{\eps\gamma}+\frac{\mathrm{e}^a}{\gamma}) \gamma t_i$.
Further, we have~$\Length_v(\sigma_{\alpha,i})\geq (1+\frac{\mathrm{e}^a}{\eps\gamma}) \gamma t_i$ by construction.
Combining these two inequalities, we obtain
\[
\frac{\Length^{D_i}_v(\sigma_{i})}{\Length^{D_i}_v(\sigma_{\alpha,i})}
=\frac{\gamma t_i+\Length_v(\sigma_i)}{\gamma t_i+\Length_v(\sigma_{\alpha,i})}
\leq \frac{2+\frac{\mathrm{e}^a}{\eps\gamma}+\frac{\mathrm{e}^a}{\gamma}}{2+\frac{\mathrm{e}^a}{\eps\gamma}}
< 1+\frac{2\eps+\frac{\mathrm{e}^a}{\gamma}}{2+\frac{\mathrm{e}^a}{\eps\gamma}}
=1+\eps \ .
\]
Summing over all vertices yields
\[
\Length^{D_i}(\sigma_i)=
\sum_{v\in V_i^*}w_v (\gamma t_i+ \Length_v(\sigma_i))
< (1+\eps) \sum_{v\in V_i^*} w_v (\gamma t_i+\Length_v(\sigma_{\alpha,i}))
=(1+\eps) \Length^{D_i}(\sigma_{\alpha,i}) \ .
\]
Using that~$\sigma_{\alpha,i}$ is an~$\alpha$-approximation for~$I_i$ completes the proof.
\end{proof}

With these lemmata at hand, Lemma~\ref{lem:euclidean-red1} easily follows.

\begin{proof}[Proof of Lemma~\ref{lem:euclidean-red1}]
Lemma~\ref{lem:euclidean-red1-sub2} implies that in the concatenation of~$\sigma_1, \dots, \sigma_q$, the sequence~$\sigma_i$ begins after a total length of at most~$\gamma t_i$ for each~$i \in [q]$. Consequently, by Lemma~\ref{lem:euclidean-red1-sub2} again, the expected total latency of the concatenated sequence, as a solution to~$I$, is bounded from above by the sum of the left-hand side of the inequality in Lemma~\ref{lem:euclidean-red1-sub2} over all~$i \in [q]$.
Thus, applying this inequality, taking the expectation, and subsequently using Lemma~\ref{lem:euclidean-red1-sub1} yields
\begin{align*}
    \mathbb{E}\left[\Length(\sigma)\right]
    &\leq \mathbb{E}\left[\sum_{i=1}^q \gamma t_i+\Length(\sigma^*_i)\right]\\
    &\leq \mathbb{E}\left[\sum_{i=1}^q \Length^{D_i}(\sigma^*_i)\right]\\
    &\leq\alpha(1+\eps^*)  \mathbb{E}\left[\sum_{i=1}^q\Length^{D_i}(\sigma^*_i)\right]\\
    &\leq\alpha(1+\eps^*)(1+\eps^*)  \Length(\sigma^*) \ .    
\end{align*}
Finally, setting~$\eps=\sqrt{1+ \frac{\eps^*}{\alpha}}-1$ completes the proof. 
\end{proof}

Recall that the value~$b$ was drawn uniformly at random from the interval~$[0,a]$.
Thus, the partition by time points~$t_i$ is random. 
We note, however, that the algorithm can be derandomized using the same techniques as in~\cite{sitters2021polynomial}, i.e., by enumerating all partitions.


\medskip

\textbf{Reducing the~$\delta$-Bounded Expanding Search Problem to the~$\delta$-Bounded Expanding Search Problem with Binary Weights.}

We show the following lemma.
A similar statement has been proven by Sitters~\cite{sitters2021polynomial} for the pathwise search problem.

\begin{lemma}[See~\cite{sitters2021polynomial}, Lemma 2.10]
\label{lem:rounding}
Consider any class~$\mathcal{C}$ of graphs with edge lengths and any constants~$\alpha>1$,~$\delta>0$, and~$\eps\in(0,1]$.
If there exists a polynomial-time~$\alpha$-approximation algorithm for the~$\delta$-bounded expanding search problem with binary weights
then there exists a polyno\-mi\-al-time~$(\alpha+\eps)$-approximation algorithm for the~$\delta$-bounded expanding search problem.
\end{lemma}

\begin{proof}
We begin by defining a new weight function~$w'$ that ensures each vertex weight is polynomially bounded, while sacrificing only a factor of~$(1+\varepsilon)$ in the approximation guarantee compared to the original weights. Specifically, we set~$w_v' = \lfloor w_v / M \rfloor$, where~$M = \smash{\frac{\varepsilon W}{n + n^2 \delta}}$ and~$W = \max_v w_v$.
Observe that~$w_v' \leq \lfloor W / M \rfloor \in O_\varepsilon(n^2)$, ensuring that the weights~$w'$ are polynomially bounded. Let~$\opt'$ denote the value of an optimal solution for the rounded instance. Note that~$M \cdot \opt' \leq \opt$ holds.
Given an~$\alpha$-approximation for the rounded instance, we apply this solution to the original instance with weights~$w$. Let~$\Length^D_v$ denote the ($D$-delayed) latency of vertex~$v$ in this solution.
Then, we have
\[ \sum_{v \in V^*} w_v \Length^D_v \leq M \sum_{v \in V^*} (w_v' +1) \Length^D_v \leq M \alpha \opt' + M \sum_{v \in V^*} \Length^D_v \leq \alpha \opt + M \sum_{v \in V^*} \Length^D_v \ .\]

Since the instance is~$\delta$-bounded,~$\delta D$ is an upper bound on the distance between any two vertices.
Hence, we have~$\Length^D_v \leq D + n \delta D = (1+ n \delta) D$ for any~$v \in V^*$.\
Therefore, we obtain
\[ M \sum_{v \in V^*} \Length^D_v \leq M n (1 + n \delta) D = \varepsilon W D \leq \varepsilon \opt \ .\]
Combining the two inequalities, we obtain 
\[ \sum_{v \in V^*} w_v \Length^D_v \leq \alpha \opt + \varepsilon \opt \leq (1 + \varepsilon) \alpha \opt \ .\]

After rounding the instance to obtain an equivalent one with weights in~$\{0,1\}$, each vertex~$v$ with weight~$w > 1$ can be replaced by a clique containing~$w$ copies of~$v$. In this clique, each vertex has a weight of~$1$, and the edges within the clique have a length of~$0$. This reduction can be performed in polynomial time, as the weights are polynomially bounded.
\end{proof}

Combining \Cref{lem:euclidean-red1} and \Cref{lem:rounding} imply \Cref{lem:reduction:weighted-01}.
Finally, \Cref{thm:weighted} follows from \Cref{lem:0-1-weights} and \Cref{lem:reduction:weighted-01}.

\section{The Euclidean Case}
\label{sec:ptas}

In this section, we show the following theorem.

\begin{theorem}\label{thm:PTAS}
Let~$d\in\mathbb{N}$ be a constant. On~$d$-dimensional Euclidean graphs, there exists a polynomial-time approximation scheme (PTAS) for the expanding search problem.
\end{theorem}

Our approach consists of three steps. 
The first two steps involve reductions inspired by Sitters~\cite{sitters2021polynomial}. 
In the first step, we reduce the expanding search problem to the~$\delta$-bounded expanding search problem with weights in~$\{0, 1\}$, as described in the previous section.
In the second step, we further reduce the~$\delta$-bounded problem to another problem, referred to as the~$\kappa$-segmented expanding search problem, for some constant~$\kappa \in \N$, while retaining weights in~$\{0, 1\}$. 
Finally, in the third step, we design a PTAS for solving the~$\kappa$-segmented expanding search problem in the Euclidean case, building upon ideas by Arora~\cite{Arora98} and Sitters~\cite{sitters2021polynomial}.

The~$\delta$-bounded expanding search problem was formally introduced in \Cref{sec:weighted}.
Next, we define the~$\kappa$-segmented expanding search problem.
Its input is as for the expanding search problem; in particular, the parameter~$\kappa$ is not part of the input.
The solution also contains~$\kappa+1$ additional time steps~$\smash{0 = t^{(0)} \leq t^{(1)} \leq \dots \leq t^{(\kappa)}}$.
For each vertex~$v \in V$, its rounded search time is given by
$$
\bar{\Length}_v(\sigma) =
\min \left\{ t^{(i)} : 0 \leq i \leq \kappa,\, \Length_v(\sigma) \leq t^{(i)} \right\} \ . 
$$
The objective is to minimize the total rounded search time,~$\smash{\bar{\Length}(\sigma) = \sum_{v \in V^*} w_v \bar{\Length}_v(\sigma)}$. Given~$\sigma$, we call its maximal prefix of total length at most~$\smash{t^{(1)}}$ segment~$1$ (of~$\sigma$). For~$i\geq 2$, we call the part of the maximal prefix of total length at most~$\smash{t^{(i)}}$ that is not part of segment~$i'$ for any~$i'<i$ segment~$i$ (of~$\sigma$).

\subsection{Reducing the~$\delta$-Bounded Expanding Search Problem to the~$\kappa$-Segmented Expanding Search Problem}

The following lemma can be proven analogously to a lemma of Sitters~\cite{sitters2021polynomial}.

\begin{lemma}[See~\cite{sitters2021polynomial}, Lemma 2.14]
\label{lem:ESP-PTAS:reduction-segmented-ESP}
Consider any class of graphs with edge lengths~$\mathcal{C}$, any class of weights~$\mathcal{W}$, and any constants~$\alpha>1$,~$\delta>0$, and~$\eps\in(0,1]$.
If there exists a polynomial-time~$\alpha$-approximation algorithm for the~$\kappa$-segmented expanding search problem for each constant~$\kappa\in\N$ on~$\mathcal{C}$ with weights~$\mathcal{W}$, then there exists a polynomial-time~$(\alpha+\eps)$-approximation algorithm for the~$\delta$-bounded expanding search problem~$\mathcal{C}$ with weights~$\mathcal{W}$.
\end{lemma}
\begin{proof}
    As shown in the proof of Lemma~\ref{lem:euclidean-red1-sub2}, there exists a~$(1 + \varepsilon)$-approximate solution that completes within time~$(1 + \delta)D$, where~$D$ denotes the delay of the given instance.
    Consider time points~$\smash{t^{(q)} = (1 + \varepsilon)^q D}$ for~$q = 1, 2, ..., \kappa$, where~$\kappa$ is such that~$(1 + \varepsilon)^\kappa \geq (1 + \delta)$.
    Note that~$\kappa$ is constant.
    Now an~$\alpha$-approximate solution for the~$\kappa$-segmented version of the instance can be converted into a~$((1 + \varepsilon)^2 \alpha)$-approximate solution for the original instance.
    The first multiplicative loss of~$(1 + \varepsilon)$ in the objective is due to the first statement of this proof, and the second multiplicative loss of~$(1 + \varepsilon)$ is since we have rounded search times, and two consecutive rounded search times satisfy~$\smash{(1 + \varepsilon) t^{(q)} =  t^{(q+1)}}$.
\end{proof}


\subsection{A PTAS for the~$\kappa$-Segmented Expanding Search Problem in the Euclidean Case}
\label{sec:PTAS-euclidean}
Sitters~\cite{sitters2021polynomial} observed that, on Euclidean graphs, the QPTAS for the pathwise search problem~\cite{AroraK03} (which builds on the well-known PTAS by Arora for TSP~\cite{Arora98}) can be transformed into a PTAS for the segmented version of the pathwise search problem.
In this section, we note that a similar, adapted approach provides a PTAS for the Euclidean segmented expanding search problem with weights in~$\{0,1\}$.
We focus on the two-dimensional case, though extending the approach to the~$d$-dimensional case for constant~$d$ is straightforward.
While the following description is self-contained, familiarity with Arora's PTAS~\cite{Arora98} may still be beneficial.

\medskip
{\bf Setup.}
The core of our PTAS for the segmented expanding search problem is the dynamic-pro\-gramming procedure.
However, several preprocessing steps are performed before invoking this procedure.
First, consider the smallest axis-aligned square that contains the root and all weight-$1$ vertices from the input, denoted as~$S_0$, with side length~$\lambda_0$.
Note that~$\lambda_0$ provides a lower bound on the cost of an optimal solution.
However, an optimal solution is not necessarily entirely contained within~$S_0$ since it may utilize a weight-$0$ vertex outside the square as a Steiner vertex.
Therefore, we enlarge~$S_0$ from its center by a factor of~$3n^2 + 1$, yielding a new square~$S$ with side length~$\lambda = (3n^2 + 1)\lambda_0$.
This scaling factor is chosen to include all vertices whose distance from~$S_0$ is at most~$\sqrt{2}n^2\lambda_0$. 
Importantly, there exists an optimal solution that is entirely contained within~$S$, as a trivial upper bound on the cost of the optimal solution is~$\sqrt{2}n^2\lambda_0$, obtained by connecting all weight-$1$ vertices to~$r$, since the distance from~$r$ to any other vertex is bounded by~$\sqrt{2}\lambda_0$.
Thus, we can safely disregard all input vertices located outside~$S$.

\paragraph{Round the instance}
We place a grid of granularity~$g\in\Theta(\varepsilon \lambda / n^4)$ within~$S$ and move each vertex to its closest grid point. 
Note that, in this process, multiple vertices may end up being relocated to the same grid point.
As shown in the literature~\cite{AroraK03}, any solution for the rounded instance can be transformed into a solution for the original instance with an additional cost of~$O(\varepsilon) \cdot \textsc{Opt}$ in the objective-function value.
This additional cost arises because the movement of each vertex incurs an extra cost of~$O(\varepsilon \lambda / n^3)$, which can be charged to the objective function since the objective is~$\Omega(\lambda / n^2)$ by the construction of~$S$.

\paragraph{Build random quadtree}
We first obtain an even larger square from~$S$ by enlarging it by an additional factor of~$2$ from its center.
Subsequently, we shift this square to the left by a value~$a$ chosen uniformly at random from~$\{-\lambda/2, -\lambda/2+g, \dots, \lambda/2-g, \lambda/2\}$ and to the top by a value~$b$, which is chosen independently of~$a$ from the same set.
Note that, regardless of the values of~$a$ and~$b$, the resulting square~$S'$ always contains~$S$.

We partition~$S'$ into four equal-sized squares, which are then recursively partitioned in the same manner until each square contains only a single grid point with at least one vertex.
This recursive partitioning naturally gives rise to a quadtree, where each square (referred to as a cell in the following) corresponds to a node in the tree.
A node is designated as a child of another node if its square is one of the four smaller squares within the parent node's square.
The quadtree is rooted at the node corresponding to~$S'$.
Due to the rounding step, the minimum distance between any two vertices not located at the same grid point is~$\Theta(\eps \lambda/n^4)$, implying that the quadtree has a depth of~$O(\log \lambda)$.

\paragraph{Derandomization}
We remark that the randomization is introduced solely for a more straightforward analysis.
In fact, we can derandomize our algorithm in the same manner as Arora's PTAS and its variants:
Specifically, we try all polynomially many possible values for the random variables~$a$ and~$b$ and output the cheapest solution obtained among these.

\medskip
{\bf Portal-Respecting Solutions and the Structural Result.}
We employ a dynamic programming approach to obtain the desired solution.
The set of solutions over which the dynamic-programming procedure optimizes consists of so-called portal-respecting solutions.
These solutions are restricted to crossing cell boundaries only at predefined locations called portals, and they do so at most~$O(1/\eps)$ times per cell in total.
For each cell, we place~$\Theta(\log n/\eps)$ equidistant portals along each side of the cell, spanning from corner to corner and including the corners.
Furthermore, each cell inherits all portals from its ancestor cells in the quadtree.

The following structural result asserts that restricting to portal-respecting solutions incurs only a~$(1 + O(\eps))$ factor increase in the cost.

\begin{lemma}\label{thm:Arora-structure}
    With constant probability (over the random placement of the quadtree), there exists a~$(1+O(\eps))$-approximate portal-respecting solution.
\end{lemma}

The result can be proved in exactly the same way as in~\cite{AroraK03} by applying Arora's structural result~\cite{Arora98} to each segment.
Although~\cite{AroraK03} considers the pathwise version of our problem, this difference does not affect the proof.

\medskip
{\bf Further Setup.}
Before we describe the dynamic program, we need two additional setup steps.

\paragraph{Additional rounding}
Since we are going to guess the lengths of parts of the solution, we assume that the distance between any two relevant points (i.e., vertices or portals) is a polynomially bounded integer. To see that this is only at the loss of another~$1+O(\eps)$ factor, recall that~$\lambda_0$ is a lower bound on the optimal objective-function value, and observe that the length of the cheapest portal-respecting solution can be viewed as a sum of~$\smash{n^{O(1)}}$ such relevant distances. At the loss of a~$1+O(\eps)$ factor in the objective-function value, we can thus afford to round all these distances to multiples of~$\smash{\eps\lambda_0/n^{O(1)}}$. Since all these distances are bounded by~$\lambda$, which is within~$\smash{n^{O(1)}}$ of~$\lambda_0$, we obtain~$\smash{n^{O(1)}}$ possible distances, and the claim follows by rescaling. 

\paragraph{Enumeration of segment lengths}
It will be useful to know the rounded latencies~$\smash{t^{(1)},\dots,t^{(\kappa)}}$ before running the dynamic-programming procedure.
By our rounding procedure, we know that there are only~$\smash{n^{O(1)}}$ options for each of these~$O(1)$ lengths, meaning that there are~$\smash{n^{O(1)}}$ combinations of different latencies for each of the segments, which we can guess.

\medskip
{\bf Dynamic Program.}
For each cell~$z$ of the quadtree, we additionally determine the following pieces of information relevant to the other quadtree cells (which is reflected in the fact that there is a DP entry for each combination).
Specifically, for each segment~$i \in [\kappa]$, we determine:

\begin{enumerate}[label=\arabic*),leftmargin=*]
	\item the total length~$\ell_i$ of segment~$i$ within the cell,
	\item the number~$m_i$ of times the segment crosses the boundary of the cell, and for each~$j \in [m_i]$ of these crossings, a \emph{type}~$\tau_{i,j}$ for the~$j$-th such crossing, containing:
	\begin{itemize}
    	\item the portal~$p_{i,j}$ at which the cell is intersected, and
    	\item whether the segment enters or leaves the cell at~$p_{i,j}$.
	\end{itemize}
\end{enumerate}

It is important to note that for each of these parameters, there are only a polynomial number of possible options.
Specifically, we can assume that~$m_i$ is at most~$O(1/\eps)$, and the number of possible types for each crossing is at most~$O(\log n/\eps)$.
This implies that there are only a polynomial number of DP entries.
Each DP entry~$$\mathrm{DP}\left[z,\left(\ell_i,(\tau_{i,j})_{j\in[m_i]}\right)_{i \in [\kappa]}\right]$$ is designed to store the cost of the cheapest portal-respecting solution confined to the corresponding cell. This solution must adhere to the constraints imposed by the chosen parameters and visit all the vertices within the cell.
It is possible that such a solution does not exist—for example, if the cell contains no root but contains other vertices, and no segment ever enters the cell.
In such cases, the cost is~$\infty$.
However, if a solution does exist, the cost refers to the sum of the latencies for all vertices in the cell, which are determined by the segment visiting each vertex.

With this definition, the entry~$\smash{\mathrm{DP}\bigl[z_0,\bigl(t^{(i)} - \sum_{i' < i} t^{(i')}, (\cdot)\bigr)_{i \in [\kappa]}\bigr]}$ is intended to store the cost of the optimal portal-respecting solution, where~$z_0$ is the root of the quadtree, and~$(\cdot)$ denotes the empty tuple.
Using standard techniques, the actual solution can be reconstructed from these entries.

The entries of the DP can be filled in a bottom-up fashion.
Specifically, an entry~$\smash{\mathrm{DP}\bigl[z,\left(\ell_i,(\tau_{i,j})_{j\in[m_i]}\right)_{i \in [\kappa]}\bigr]}$ where~$z$ is a leaf of the quadtree can be computed easily.
If the cell does not contain the root but contains at least one other vertex (with all vertices located at a common point), and there are no incoming crossings, we set the DP entry to~$\infty$.
Otherwise, we make a guess to determine which segment~$i_0$ and which incoming crossing~$j_0 \in [m_i]$ connect to the vertices.
If the cell contains a vertex or the root (in which case we set~$i_0 = 0$), we guess which of the outgoing crossings~$j > j_0$ of segment~$i_0$ and which outgoing crossing from later segments connect to the vertex. 

We then guess a one-to-one correspondence between the remaining incoming and outgoing crossings for each segment, ensuring that each incoming crossing is paired with an outgoing crossing of a later index.
The corresponding portals are then connected.
If no such correspondence exists for any segment (e.g., because the number of remaining incoming and outgoing crossings is unequal), we discard this guess.
Similarly, we discard the guess if the resulting length of segment~$i$ within~$z$ does not match~$\ell_i$.
Among the valid options, we store the lowest cost in the DP entry.
If no valid solution is found, we set the cost to~$\infty$.
Note that we are only considering a polynomial number of guess combinations.

To compute a DP entry
$\smash{\mathrm{DP}\bigl[z,(\ell_i,(\tau_{i,j})_{j\in[m_i]})_{i\in[\kappa]}\bigr]}$
for a non-leaf node~$z$, we use previously computed entries for the children of~$z$:
\begin{align*}
&z^{\text{TL}} \quad \text{(top-left)}, & &z^{\text{TR}} \quad \text{(top-right)},\\
&z^{\text{BL}} \quad \text{(bottom-left)}, &   &z^{\text{BR}} \quad \text{(bottom-right)}.
\end{align*}
For each segment~$i \in [\kappa]$, we first guess nonnegative integers~$\ell_i^{\text{TL}}$,~$\ell_i^{\text{TR}}$,~$\ell_i^{\text{BL}}$, and~$\ell_i^{\text{BR}}$ such that the total length of the segment is preserved:~$\ell_i = \ell_i^{\text{TL}} + \ell_i^{\text{TR}} + \ell_i^{\text{BL}} + \ell_i^{\text{BR}}$.

For the crossings, note that~$(\tau_{i,j})_{i\in[\kappa], j\in[m_i]}$ already specify the crossings (and their types) for the sides of the children cells that align with the sides of~$z$ (the outer boundaries). However, these do not determine the crossings along the inner boundaries between the children cells. Therefore, we guess the crossings and their types for the inner boundaries and determine the order in which these crossings occur, ensuring consistency with the ordering of the crossings along the outer boundaries.
These guesses result in the crossings
\begin{align*}
(\tau_{i,j}^\text{TL})_{i \in [\kappa], j \in [m_i^{\text{TL}}]}, (\tau_{i,j}^\text{TR})_{i \in [\kappa], j \in [m_i^{\text{TR}}]}, (\tau_{i,j}^\text{BL})_{i \in [\kappa], j \in [m_i^{\text{BL}}]}, \text{ and } (\tau_{i,j}^\text{BR})_{i \in [\kappa], j \in [m_i^{\text{BR}}]}
\end{align*}
for the children cells. 
Note that a single guessed crossing of an inner boundary leads to two crossings (one outgoing and one incoming) for each of the children cells.
Finally, we store the minimum value of the following sum in the DP entry~$\mathrm{DP}[z,(\ell_i,(\tau_{i,j})_{j \in [m_i]})_{i \in [\kappa]}]$:
\begin{align*}
&\mathrm{DP}\left[z^{\text{TL}},\left(\ell_i^{\text{TL}},\left(\tau_{i,j}^{\text{TL}}\right)_{j\in[m_i^{\text{TL}}]}\right)_{i\in[\kappa]}\right]+
\mathrm{DP}\left[z^{\text{TR}},\left(\ell_i^{\text{TR}},\left(\tau_{i,j}^{\text{TR}}\right)_{j\in[m_i^{\text{TR}}]}\right)_{i\in[\kappa]}\right]\\
+\;&\mathrm{DP}\left[z^{\text{BL}},\left(\ell_i^{\text{BL}},\left(\tau_{i,j}^{\text{BL}}\right)_{j\in[m_i^{\text{BL}}]}\right)_{i\in[\kappa]}\right]+
\mathrm{DP}\left[z^{\text{BR}},\left(\ell_i^{\text{BR}},\left(\tau_{i,j}^{\text{BR}}\right)_{j\in[m_i^{\text{BR}}]}\right)_{i\in[\kappa]}\right] \ .
\end{align*}
This is the minimum cost obtained through these guesses. Again, note that only polynomially many combinations of guesses are considered. This completes the description of the dynamic program.

By the correctness of this dynamic program and \Cref{thm:Arora-structure}, we obtain a PTAS for the~$\kappa$-segmented expanding search problem on Euclidean graphs with binary weights. Using \Cref{lem:ESP-PTAS:reduction-segmented-ESP} and \Cref{lem:reduction:weighted-01}, we then obtain \Cref{thm:PTAS}.

\section{Hardness of Approximation}
\label{sec:hardness}
This section is dedicated to the following theorem. 

\begin{theorem}\label{thm:hardness}
    There exists some constant~$\eps>0$ such that there is no polynomial-time~$(1+\eps)$-approximation algorithm for the expanding search problem, unless~$\P=\NP$.
\end{theorem}

The hardness result for the expanding search problem follows from a reduction from a variant of the Steiner tree problem, which is defined as follows. 
Given a graph~$G=(V,E)$ with non-negative edge lengths and a set~$T\subseteq V$ of vertices, the so-called terminals, the Steiner tree problem on graphs asks for a minimum-length tree that is a subgraph of~$G$ and that
contains all vertices in~$T$.
The variant that we consider and use is the so-called \textsc{SteinerTree(1,2)}, short \textsc{ST(1,2)}, where~$G$ is a complete graph, and all edge lengths are either~$1$ or~$2$.
Bern and Plassmann~\cite{bern1989steiner} showed the following theorem.

\begin{theorem}[Theorem 4.2 in~\cite{bern1989steiner}] \label{thm:steiner12}
    \textsc{SteinerTree(1,2)} is~$\mathsf{MaxSNP}$-hard.
\end{theorem}

It was shown in~\cite{arora1998proof} that no polynomial-time approximation scheme exists for any~$\mathsf{MaxSNP}$-hard problem unless~$\P=\NP$. Hence, there exists some constant~$\rho>0$ such that there is no polynomial-time~$(1+\rho)$-approximation algorithm for \textsc{ST}(1,2), unless~$\P=\NP$.
We use this to show the hardness result for the expanding search problem.

The main idea of the proof of \Cref{thm:hardness} is as follows. Given a~$\beta$-approxima-tion algorithm~$\alg'$ for the expanding search problem for any~$\beta>1$, we construct a~$\gamma$-approximation algorithm~$\alg$ for \textsc{ST}(1,2) with~$\gamma < 1+\rho$. With the approximation hardness of \textsc{ST}(1,2), this contradicts the existence of a~$\beta$-approximation algorithm~$\alg'$ for the expanding search problem for any~$\beta>1$.

\medskip
{\bf Construction of the Expanding Search Instance.}
Let~$I=(G,T,(\ell_e)_{e\in E})$ be an instance of \textsc{ST}(1,2) on the undirected complete graph~$G=(V, E)$ with terminal set~$T\subseteq V$ with~$|T| \geq 2$ and edge lengths~$\ell_e\in \{1,2\}$ for all~$e\in E$. 
We construct an instance~$I'=(G',(w_v)_{v\in V'},(\ell'_e)_{e\in E'},r)$ of the expanding search problem as follows.
The graph~$G'=(V',E')$ consists of~$k$ copies~$G_1,\dots,G_k$ of~$G$ where all vertices are connected to an additional vertex~$r$, i.e., the root vertex of the the expanding search problem instance. 
The number~$k$ of copies will be determined later.
All edges within some copy~$G_i$ are assigned the same length as in the original instance~$G$. 
Edges incident to~$r$ have length~$a=2(|T| -1)$. 
Finally, all vertices corresponding to terminals in the original instance (later called terminal vertices) have weight~$1/|T|$ while all other vertices are assigned weight~$0$. 
Note that by this choice of weights, each copy of~$G$ has a total weight of 1.
We refer to \Cref{fig:construction-esp} for an illustration of the construction.

\begin{figure}[t]
\scriptsize
\centering
\begin{subfigure}[b]{.48\textwidth}
\begin{center}
\begin{tikzpicture}[xscale=0.5,yscale=0.5,shorten > = 0pt]
	\foreach \i in {1,...,5} {
    \node[state] (n\i) at ({360/5 * (\i - 1)}:3) {};
	}
	\node[state,\myred] (test) at ({360/5 * (2 - 1)}:3) {};
	\node[state,\myred] (test) at ({360/5 * (4 - 1)}:3) {};
	\node[state,\myred] (test) at ({360/5 * (1 - 1)}:3) {};
	\foreach \i in {1,...,5} {
	    \foreach \j in {\i,...,5} {
	        \ifnum\i<\j
	            \draw[thick, \myyellow] (n\i) -- (n\j);
	        \fi
	    }
	}
	\draw[thick, \myblue] (n1) -- (n3);
	\draw[thick, \myblue] (n2) -- (n3);
	\draw[thick, \myblue] (n4) -- (n3);
	\draw[thick, \myblue] (n5) -- (n2);
\end{tikzpicture}
\end{center}
\caption{}
\label{fig:construction-esp-1}
\end{subfigure}
\begin{subfigure}[b]{.48\textwidth}
\begin{center}
\begin{tikzpicture}[xscale=0.2,yscale=0.2,shorten > = 0pt]
	\begin{scope}[xshift=15cm]
			\foreach \i in {1,...,5} {
	    \node[state] (m\i) at ({360/5 * (\i - 1)}:3) {};
		}
		\node[state,\myred] (test) at ({360/5 * (2 - 1)}:3) {};
		\node[state,\myred] (test) at ({360/5 * (4 - 1)}:3) {};
		\node[state,\myred] (test) at ({360/5 * (1 - 1)}:3) {};
		\foreach \i in {1,...,5} {
		    \foreach \j in {\i,...,5} {
		        \ifnum\i<\j
		            \draw[thick, \myyellow] (m\i) -- (m\j);
		        \fi
		    }
		}
		\draw[thick, \myblue] (m1) -- (m3);
		\draw[thick, \myblue] (m2) -- (m3);
		\draw[thick, \myblue] (m4) -- (m3);
		\draw[thick, \myblue] (m5) -- (m2);
		\node[label=right:{$G_k$}] (Gk) at (-2,-5) {};
	\end{scope}
	\foreach \i in {1,...,5} {
    \node[state] (n\i) at ({360/5 * (\i - 1)}:3) {};
	}
	\node[state,\myred] (test) at ({360/5 * (2 - 1)}:3) {};
	\node[state,\myred] (test) at ({360/5 * (4 - 1)}:3) {};
	\node[state,\myred] (test) at ({360/5 * (1 - 1)}:3) {};
	\foreach \i in {1,...,5} {
	    \foreach \j in {\i,...,5} {
	        \ifnum\i<\j
	            \draw[thick, \myyellow] (n\i) -- (n\j);
	        \fi
	    }
	}
	\draw[thick, \myblue] (n1) -- (n3);
	\draw[thick, \myblue] (n2) -- (n3);
	\draw[thick, \myblue] (n4) -- (n3);
	\draw[thick, \myblue] (n5) -- (n2);
	\node[label=right:{$G_1$}] (G1) at (-2,-5) {};
	\node[label=right:{\large~$\dots$}] (dots) at (5,0) {};
	\node[state,label=above:{$s$}] (s) at (7.5,7) {};
	\foreach \j in {1,...,5} {
		            \draw[thick, \mygreen] (s) -- (m\j);
		    }
    \foreach \j in {1,...,5} {
            \draw[thick, \mygreen] (s) -- (n\j);
    }
\end{tikzpicture}
\end{center}
\caption{}
\label{fig:construction-esp-2}
\end{subfigure}
\caption{(a)~Instance~$I$ for \textsc{ST}(1,2). Terminal vertices are shown in \textcolor{\myred}{red}.
\textcolor{\myblue}{Blue} edges have length~$1$ and \textcolor{\myyellow}{orange} edges have length~$2$.
(b)~Instance~$I'$ for the expanding search problem constructed from~$I$.
\textcolor{\myred}{Red} vertices have weight~$\smash{\frac{1}{3}}$ and black vertices have weight~$0$.
\textcolor{\myblue}{Blue} edges have length~$1$, \textcolor{\myyellow}{orange} edges have length~$2$, and \textcolor{\mygreen}{green} edges have length~$a=2(|T|-1)$.}
\label{fig:construction-esp}
\end{figure}

\medskip
We make the following observation. Any feasible Steiner tree for~$I$ consists of at least~$|T|-1$ edges, each having a length of 1 or 2. 
Hence, we can lower-bound the length of an optimal Steiner tree for~$I$ by~$|T|-1$.
However, since~$G$ is a complete graph, choosing a spanning tree that only uses edges from~$E[T]$ gives an upper bound for the length of an optimal Steiner tree of~$2(|T|-1)$.
This yields
\begin{align} \label{eq:ah-bound-a}
    \opt_\text{ST}(I)\leq a\leq 2 \,\opt_\text{ST}(I)\ , 
\end{align}
where~$\opt_\text{ST}(I)$ denotes the length of an optimal Steiner tree solution on~$I$.

To show the main result, we make some assumptions on the expanding search pattern~$\sigma_\alg$ obtained from the~$\beta$-approximation algorithm~$\alg'$ for the expanding search pattern on instance~$I'$.
For this manner, we call~$\sigma_\alg$ \emph{structured} if each copy~$G_i$ is connected to the root by precisely one edge in~$\sigma_\alg$ and all edges belonging to some copy~$G_i$ or connecting~$G_i$ to the root are consecutive within~$\sigma_\alg$.
The following lemma states that we can always obtain a structured expanding search pattern from the~$\beta$-approximation~$\sigma_\alg$ in polynomial time that maintains the approximation ratio.
Note that, with~$a$ chosen large enough, an analogous statement would be much easier to show for the pathwise search problem since the searcher would need to cross edges of length~$a$ more often if edges from the same copy are not traversed consecutively.
\begin{lemma}\label{lem:ah-structured}
    Given an expanding search pattern~$\sigma_\alg$, a structured search pattern~$\sigma'_\alg$ can be computed in polynomial time such that~$L(\sigma'_\alg) \leq L(\sigma_\alg)$.
\end{lemma}
\begin{proof}[Proof of \Cref{lem:ah-structured}]
Nothing is left to show if~$\sigma_\alg$ is structured, so we may assume that~$\sigma_\alg$ is not structured.
First, suppose that some copy~$G_i$ is connected to the root by more than one edge. Let~$e=(r,v)$ be the first edge that connects~$G_i$ to the root in~$\sigma_\alg$. Then, we can replace any other edge~$(r,w)$ with~$w\in V[G_i]$ by the edge~$(v,w)$ of length at most~$2\leq a$. 

Hence, if~$\sigma_\alg$ is not structured, we may assume that this is due to the existence of some copy~$G_i$ for which not all edges belonging to~$G_i$ or connecting~$G_i$ to the root are consecutive within~$\sigma_\alg$.
We write~$\sigma_\alg$ as a concatenation of (consecutive) subsequences
\begin{align*}
\sigma_\alg = \sigma_1 + \sigma_2 + \sigma_3 + \dots + \sigma_{2s} + \sigma_{2s+1} \ , 
\end{align*}
for some~$s > 1$
such that the subsequences with even index~$\sigma_2,\sigma_4,\dots,\sigma_{2s}$ are the inclusion-wise maximal subsequences of~$\sigma_\alg$ consisting only of edges belonging to~$G_i$ or connecting~$G_i$ to the root in the order as they appear in~$\sigma_\alg$.
The subsequences~$\sigma_1,\sigma_3,\dots,\sigma_{2s+1}$ with odd index are the inclusion-wise maximal subsequences of the remaining edges where we allow that~$\sigma_1$ and~$\sigma_{2s+1}$ are empty, but all other subsequences with odd index are non-empty. 
For some subsequence~$\hat{\sigma}$ of~$\sigma_\alg$ we denote by~$\ell(\hat{\sigma})=\sum_{e\in \hat{\sigma}}\ell_e$ the length of that subsequence and by~$t(\hat{\sigma})$ the number of terminal vertices that~$\hat{\sigma}$ connects to the rooted tree that previous subsequences have already explored. 
Then, we can define the \emph{ratio} of a subsequence as~$r(\hat{\sigma})= \ell(\hat{\sigma})/t(\hat{\sigma})$.

\begin{claim}\label{cl:t>1}
    Without loss of generality, we can assume that~$t(\sigma_j)\geq 1$ for all values~$j\in\{2,3,\dots,2s\}$.
\end{claim}
\begin{proof}[Proof of the claim]
    We assume that there exists a value~$j\in\{2,3,\dots,2s\}$ such that~$t(\sigma_j)=0$.
    Then, we can swap the positions of~$\sigma_j$ and~$\sigma_{j+1}$ and continue with the newly obtained expanding search sequence with fewer subsequences. By this, we only improve the total latency since no exploration of any terminal vertex has been postponed.
    \renewcommand{\qedsymbol}{$\triangleleft$}
\end{proof}

With \Cref{cl:t>1} the ratio~$r(\sigma_j)$ is well defined for all~$j\in\{2,3,\dots,2s\}$. 

\begin{claim}\label{cl:r>=2}
    We have~$r(\sigma_2)\geq 2$.
\end{claim}
\begin{proof}[Proof of the claim]
    The statement follows from the fact that the very first edge in~$\sigma_2$ has length~$a=2(|T|-1)$ and \Cref{cl:t>1} saying that there exists at least one more subsequence~$\sigma_{2l}$ for some~$l>1$ that connects at least one terminal vertex.
    Therefore,~$\sigma_2$ can connect at most~$|T|-1$ terminal vertices for which at least~$|T|-1$ many edges are needed.
    Each of those edges has a length of~$1,2,$ or~$a$, where precisely one edge has length~$a$.
    Hence, the minimum possible ratio for~$\sigma_2$ is
\begin{align*}
\frac{a+(|T|-2)}{|T|-1}=3-\frac{1}{|T|-1}\geq 2 \ ,
\end{align*}
where we use~$|T|\geq 2$.
\renewcommand{\qedsymbol}{$\triangleleft$}
\end{proof}

\begin{claim}\label{cl:r<=2}
    Without loss of generality, we can assume that~$r(\sigma_{2s})\leq 2$.
\end{claim}
\begin{proof}[Proof of the claim]
    Assume that this was not the case, so~$r(\sigma_{2s})> 2$. Then, we can make the following changes to decrease that ratio without increasing the total latency of the initial expanding search pattern~$\sigma_\alg$. 
    Let~$\sigma_{2s}=(e_1,\dots,e_m)$. If~$e_m$ does not connect a terminal vertex, we can remove~$e_m$ from~$\sigma_\alg$.
    With~$\ell_{e_m}>0$, this decreases the ratio~$r(\sigma_{2s})$ and only decreases the latencies of any terminal vertices connected by edges in~$\sigma_{2s+1}$.    
    Hence, we may assume that~$e_m$ connects a terminal vertex.
    Now, choose the shortest subsequence~$\Bar{\sigma}=(e_p,\dots,e_{m-1},e_m)$ with~$p\in [m]$ such that~$r(\Bar{\sigma})>2$.
    This is well-defined since~$r(\sigma_{2s})>2$.
    Let~$\{e^*_1,\dots,e^*_q=e_m\}$ be the set of edges that connect a terminal vertex in the order as they appear in~$\Bar{\sigma}$.
    We claim that for any subsequence~$\Bar{\sigma}_j=(e_p,e_{p+1},\dots,e^*_j)$ with~$j\in [q]$ it holds that~$r(\Bar{\sigma}_j)>2$.
    Assume for contradiction that there exists some~$j\in [q]$ such that~$r(\Bar{\sigma}_j)\leq 2$.
    Then, let~$\Bar{\sigma}_{-j}$ be such that~$\Bar{\sigma}$ is a concatenation of~$\Bar{\sigma}_{j}$ and~$\Bar{\sigma}_{-j}$.
    However, since~$\Bar{\sigma}$ is the shortest contiguous subsequence of~$\sigma_{2s}$ that contains~$e_m$ such that~$r(\Bar{\sigma})>2$ holds, it follows that~$r(\Bar{\sigma}_{-j})\leq 2$, otherwise~$\Bar{\sigma}$ would not be minimal.
    In total, this yields
    \begin{align*}
        2
        < r(\Bar{\sigma})
        =\frac{\ell(\Bar{\sigma})}{t(\Bar{\sigma})}
        =\frac{\ell(\Bar{\sigma}_{j})+\ell(\Bar{\sigma}_{-j})}{t(\Bar{\sigma}_{j})+t(\Bar{\sigma}_{-j})}
        \leq \frac{2t(\Bar{\sigma}_{j})+2t(\Bar{\sigma}_{-j})}{t(\Bar{\sigma}_{j})+t(\Bar{\sigma}_{-j})}
        =2 \ ,
    \end{align*}
    a contradiction.
    Hence, for any subsequence~$\Bar{\sigma}_j=(e_p,\dots,e^*_j)$ with~$j\in [q]$, it holds that~$r(\Bar{\sigma}_j)>2$.
    Let~$Q$ denote the set of terminal vertices connected by~$\Bar{\sigma}$.
    We can remove all edges in~$\Bar{\sigma}$ and instead connect each vertex in~$Q$ with a single edge of length~$2$ in the same order as they had been connected in~$\Bar{\sigma}$. This strictly decreases the latency of all vertices in~$Q$ and, in particular, the latency of that terminal vertex, connected initially by~$e_m$.  
    Hence, all terminal vertices connected by some edges in~$\sigma_{2s+1}$ will also have a strictly smaller latency. Additionally, this strictly decreases the ratio~$r(\sigma_{2s})$, so we can repeat this procedure until~$r(\sigma_{2s})\leq 2$ as there are only finitely many values that this ratio can take.
    \renewcommand{\qedsymbol}{$\triangleleft$}
\end{proof}

With those three claims, we are now ready to prove \Cref{lem:ah-structured}.
    Consider the ratios~$r(\sigma_{2s-1})$ and~$r(\sigma_{2s})$ and distinguish two cases.
    First, assume~$r(\sigma_{2s-1})\geq r(\sigma_{2s})$. We then claim that swapping those two subsequences decreases the total latency. To this end, note that~$r(\sigma_{2s-1}) \geq r(\sigma_{2s})$ yields
    \begin{align}
        \frac{\ell(\sigma_{2s-1})}{t(\sigma_{2s-1})}
        &\geq \frac{\ell(\sigma_{2s})}{t(\sigma_{2s})} &&\Leftrightarrow  &
        \ell(\sigma_{2s-1}) t(\sigma_{2s})
        \geq \ell(\sigma_{2s}) t(\sigma_{2s-1}) \ . \label{eq:ah-ratio}
    \end{align}
    Swapping the positions of~$\sigma_{2s-1}$ and~$\sigma_{2s}$ causes the latency of~$t(\sigma_{2s})$ many terminal vertices to decrease by~$\ell(\sigma_{2s-1})$ while the latency of~$t(\sigma_{2s-1})$ many terminal vertices increase by~$\ell(\sigma_{2s})$. Note that the latencies of all terminal vertices connected by other subsequences remain unaffected.
    Hence, by Equivalence~(\ref{eq:ah-ratio}), the total latency of the expanding search pattern cannot increase.
    It remains to consider the second case, i.e.,~$r(\sigma_{2s-1})< r(\sigma_{2s})$, which yields~$r(\sigma_{2s-1})<2$.
    We then compare~$r(\sigma_{2s-1})$ to~$r(\sigma_{2s-2})$ and continue recursively with adjacent subsequences until we find the first pair~$\sigma_{j}$ and~$\sigma_{j-1}$ with~$j\in\{2,3,\dots,2s-2\}$ such that~$r(\sigma_{j})\geq r(\sigma_{j-1})$. This pair does exist since~$r(\sigma_{2s-1})<2$ and~$r(\sigma_2)\geq 2$.
    If this pair is found, those two subsequences can swap positions, and with an analogous computation to the one in the first case, the total latency does not increase.
    After swapping at most~$s\leq |V|$ pairs of subsequences in the same way, the desired property is established for~$G_i$.
    
    We repeat this process for each copy of~$G$ until the obtained expanding search sequence is structured. Computing the ratios and performing the swaps of the subsequences takes time polynomial in the length of the sequence; hence, this procedure has a polynomial running time.
\end{proof}

With \Cref{lem:ah-structured} at hand, we will assume from now on that~$\sigma_\alg$ is structured and that it visits the~$k$ copies of~$G$ in the order~$G_1,\dots,G_k$. Next, we define the algorithm~$\alg$ that uses the~$\beta$-approximation algorithm for the expanding search problem to obtain a solution for the original \textsc{ST}(1,2) instance~$I$.

\medskip
{\bf The Algorithm for \textsc{ST}(1,2).}
The algorithm~$\alg$ for \textsc{ST}(1,2) is defined as follows:
\begin{enumerate}[label=\arabic*),leftmargin=*]
    \item It takes an \textsc{ST}(1,2) instance~$I$ as input and constructs the corresponding expanding search problem instance~$I'$.
    \item Next, it runs the~$\beta$-approximation algorithm~$\alg'$ on~$I'$ to obtain the expanding search pattern~$\sigma_\alg$.
    \item Using~$\sigma_\alg$, the algorithm computes a corresponding structured expanding search pattern~$\sigma_\alg'$. As shown in \Cref{lem:ah-structured}, this step can be performed in polynomial time.
    \item For each copy~$G_i$,~$\alg'$ computes a Steiner tree solution~$T_i$ for the original instance~$I$ as a byproduct.
    Out of these~$k$ solutions,~$\alg$ selects the one with the minimum length, defined as~$T^* = \arg\min_{i \in [k]} \bigl\{\sum_{e \in E[T_i]} \ell_e \bigr\}$.
\end{enumerate}
The output of the algorithm is tree~$T^*$, i.e.,~$\alg(I) = \ell(T^*) = \sum_{e \in E[T^*]} \ell_e$.
\medskip

In the remainder of this section, we will denote the objective value of the algorithms~$\alg$ and~$\alg'$ on instances~$I$ and~$I'$ by~$\alg(I)$ and~$\alg'(I')$. Further, we will denote the optimal objective values for \textsc{ST}(1,2) and the expanding search problem on instances~$I$ and~$I'$ by~$\opt_\text{ST}(I)$ and~$\opt_\text{ESP}(I')$, respectively.

First, we establish the upper bound on~$\alg(I)$. To do this, let~$T_1, \dots, T_k$ represent the Steiner tree solutions that~$\alg'$ computes for the~$k$ copies~$G_1, \dots, G_k$ of the instance~$I'$.
The upper bound is derived by assuming that the expanding search pattern~$\sigma_\alg$ collects the entire weight of~$1$ for each copy~$G_i$ upon visiting the very first vertex of that copy.
This yields
\begin{align*}
    \alg'(I')
    &\geq \sum_{i=1}^k \left(ia+ \sum_{j=1}^{i-1} \ell(T_j) \right)\\
    &\geq \sum_{i=1}^k \left(ia+ \sum_{j=1}^{i-1} \ell(T^*) \right)\\
    &=  \frac{k(k+1)}{2} a + \frac{(k-1)k}{2} \ell(T^*)
\end{align*}
which is equivalent to
\begin{align}
    \alg(I)=\ell(T^*)
    &\leq \frac{2}{(k-1)k}\left( \alg'(I') - \frac{k(k+1)}{2} a \right) \ . \label{eq:ah-upper-bound-alg}
\end{align}

Next, we provide an upper bound on~$\opt_\text{ESP}(I')$. To do so, consider an optimal Steiner tree solution for the instance~$I$.

Using this optimal solution, we construct an expanding search problem solution as follows:
\begin{enumerate}[label=\arabic*),leftmargin=*]
    \item Start by visiting an arbitrary vertex~$v$ in~$G_1$.
    \item Then, traverse all edges of the optimal Steiner tree solution for~$I$.
    \item Repeat the same process for the remaining copies~$G_2, \dots, G_k$.
\end{enumerate}
By assuming that the total weight of~$1$ for each copy is only collected upon visiting the very last terminal vertex of each copy, we obtain the following upper bound on~$\opt_\text{ESP}(I')$
\begin{align}
    \opt_\text{ESP}(I')
    \leq \sum_{i=1}^k i(a+\opt_\text{ST}(I))
    =\frac{k(k+1)}{2} (a+\opt_\text{ST}(I)) \ . \label{eq:ah-upper-bound-opt}
\end{align}
Combining \Cref{eq:ah-upper-bound-alg,eq:ah-upper-bound-opt} with~$\alg'(I')\leq \beta \opt_\text{ESP}(I')$, yields
\begin{align}
    \alg(I)
    &\leq \frac{2}{(k-1)k}\left( \alg'(I') - \frac{k(k+1)}{2} a \right) \notag \\
    &\leq \frac{2}{(k-1)k}\left( \beta \opt_\text{ESP}(I') - \frac{k(k+1)}{2} a \right) \notag \\
    &\leq \frac{2}{(k-1)k}\Biggl( \beta \biggl(\frac{k(k+1)}{2} \bigl(a+\opt_\text{ST}(I)\bigr)\biggr) - \frac{k(k+1)}{2} a \Biggr) \notag\\
    &=\frac{k+1}{k-1}\left( a(\beta-1)+\beta \opt_{\text{ST}}(I)\right). \notag \\
    \intertext{With~$\beta-1>0$ and \Cref{eq:ah-bound-a}, we finally obtain }
    \alg(I)
    &\leq \frac{k+1}{k-1}(3\beta -2)\opt_{\text{ST}}(I) \ .
\end{align}

Thus, applying the~$\beta$-approximation algorithm for the expanding search problem results in a~$\gamma$-approximation algorithm for \textsc{ST}(1,2), where~$\gamma = \frac{k+1}{k-1}(3\beta - 2)$.
However, by selecting~$\beta$ and~$k$ such that~$\beta < 1 + \smash{\frac{\rho}{3}}$ and~$k > \smash{\frac{\rho + 3\beta - 1}{\rho - 3\beta + 3}}$, we achieve~$\gamma < 1 + \rho$ contradicting the approximation hardness of \textsc{ST}(1,2).
This establishes that there exists a constant~$\eps > 0$ such that no polynomial-time~$(1 + \eps)$-approximation algorithm exists for the expanding search problem unless~$\mathsf{P} = \mathsf{NP}$. This completes the proof of \Cref{thm:hardness}.

\medskip
\noindent\textbf{Acknowledgements.}
We thank Spyros Angelopoulos for fruitful discussions and pointers to earlier literature.

\bibliographystyle{siamplain}
\bibliography{references}

\end{document}